\documentclass[journal]{IEEEtran}
\usepackage[colorlinks,linkcolor=blue,anchorcolor=blue,citecolor=blue,bookmarks=true]{hyperref}  
\usepackage[T1]{fontenc}
\ifCLASSINFOpdf
\usepackage[pdftex]{graphicx} 
\DeclareGraphicsExtensions{.eps,.pdf,.jpeg,.png}
\else
\usepackage[dvips]{graphicx}
\fi
\usepackage[compress,nospace]{cite}
\usepackage[cmex10]{amsmath}
\usepackage{epstopdf,amsthm,stfloats,siunitx,amssymb,wasysym,algorithm,algorithmic,array,url, color,subfigure} 
\usepackage{footnote}
\makesavenoteenv{tabular}
\newtheorem{theorem}{Proposition}

\newtheorem{cor}[theorem]{Corollary}

\begin{document}

%% Title and Authors
\title{A New Atomic Norm for DOA Estimation With Gain-Phase Errors}

\author{Peng~Chen,~\IEEEmembership{Member,~IEEE},
	Zhimin~Chen,~\IEEEmembership{Member,~IEEE}, Zhenxin~Cao,~\IEEEmembership{Member,~IEEE}, 
%Xiaoye Wu, 
 Xianbin~Wang,~\IEEEmembership{Fellow,~IEEE}
%Yi Jin
%, Can Zhu
%\thanks{C.~Can is with NO.724 Research institute of CSIC, Nanjing 211153, China (email: ZCan1021@163.com).}
\thanks{This work was supported in part by the  National Natural Science Foundation of China (Grant No. 61801112), the Open Program of State Key Laboratory of Millimeter Waves at Southeast University (Grant No. K202029), the Natural Science Foundation of Jiangsu Province (Grant No. BK20180357)
	 the foundation of Shannxi Key Laboratory of Integrated and Intelligent Navigation (Grant  No. SKLIIN-20190204), the Fundamental Research Funds for the Central Universities (Grant  No. 2242020K40114). \textit{(Corresponding author: Peng Chen)}}
\thanks{P.~Chen and Z.~Cao are with the State Key Laboratory of Millimeter Waves, Southeast University, Nanjing 210096, China (email: \{chenpengseu,caozx\}@seu.edu.cn).}
%\thanks{X.~Wu is with the Beijing Aerospace TT\&C Technology Co., Ltd., China.}
\thanks{Z.~Chen is with the School of Electronic and Information, Shanghai Dianji University, Shanghai 201306, China (email: chenzm@sdju.edu.cn).}
\thanks{X.~Wang is with the Department of Electrical and Computer Engineering, Western University, Canada (e-mail: xianbin.wang@uwo.ca).}
%\thanks{Y.~Jin is with Xi'an branch of China Academy of Space Technology, Xi'an 710100, China (email: john.0216@163.com).} 
}

% The paper headers
\markboth{IEEE Transactions on Signal Processing}%
{Shell \MakeLowercase{\textit{et al.}}: Bare Demo of IEEEtran.cls for Journals}

\maketitle

\begin{abstract}
The problem of direction of arrival (DOA) estimation  has been studied for decades as an essential technology in enabling radar, wireless communications, and array signal processing related applications. In this paper, the DOA estimation problem in the scenario with gain-phase errors is considered, and a sparse model is formulated by exploiting the signal sparsity in the spatial domain. By proposing a new atomic norm, named as GP-ANM, an optimization method is formulated via deriving a dual norm of GP-ANM. Then, the corresponding semidefinite program (SDP) is given to estimate the DOA efficiently, where the SDP is obtained based on the Schur complement. Moreover, a regularization parameter is obtained theoretically in the convex optimization problem. Simulation results show that the proposed method outperforms the existing methods, including the subspace-based and sparse-based methods in the scenario with gain-phase errors.
\end{abstract}


\section{Introduction} \label{sec1}

The estimation problem of the direction of arrival (DOA) has been studied for decades in different applications encompassing radar, wireless communications, and array signal processing~\cite{8233171}. Traditionally, the DOA is estimated by the discrete Fourier transform (DFT)-based methods~\cite{Liu2016,Burintramart2007,Kim2015}, where the antenna arrays provide spatial samplings. The DFT-based methods realize the DOA estimation via the DFT of received signals spatially sampled by the antenna array, with its inherent sampling resolution characterized by  the \emph{Rayleigh criterion}~\cite{Wolfgang}.

To overcome the resolution limit of the Rayleigh criterion, different super-resolution methods for DOA estimation have been proposed, and the subspace-based methods have been widely used in the scenarios with multiple measurements to estimate the covariance matrix of received signals in the antenna array. For example, the multiple signal classification (MUSIC) method~\cite{ralph1986} and the estimating signal parameters via rotational invariance techniques (ESPRIT) method~\cite{roy1989}, where the MUSIC method estimate the DOA with the noise subspace but the ESPRIT uses the signal subspace. Then, the extension algorithms based on the MUSIC and ESPRIT methods are proposed in the present papers, such as Root-MUSIC method~\cite{Zoltowski1993}, space-time MUSIC method~\cite{8309290}, G-MUSIC method~\cite{7180404}, higher order ESPRIT and virtual ESPRIT~\cite{539037}, etc. Ref.~\cite{ELDAR2018} also develops a frequency estimation method in the continuous domain with sensor calibration and off-grid problems.

Recently, to further improve the DOA estimation performance, the compressed sensing (CS) methods have been proposed by exploiting the signal sparsity in the spatial domain~\cite{Xiong2018,Li2017,Yang20166,yao2011}. Ref.~\cite{yao2011,Chen2017} propose the CS-based DOA estimation methods in the multiple-input and multiple-output (MIMO) radar systems. A compressed sparse array scheme is proposed in~\cite{8386702}. However, in the CS-based methods, the dictionary matrix is formulated by discretizing the spatial domain. Consequently, the corresponding dictionary matrix is formulated using the discretized spatial angles. When the DOAs are not exactly at the discretized angles, which  introduces the \emph{off-grid} errors, and the off-grid methods have been proposed to solve this problem~\cite{6576276}. For example, the structured dictionary mismatch is considered, and the corresponding sparse reconstruction methods are proposed in~\cite{6867380}. A sparse Bayesian inference is given in~\cite{Yang2012} with the off-grid consideration. Moreover, an iterative reweighted method~\cite{Fang2014} estimates the off-grid and sparse signals jointly. In~\cite{zhu_grid-less_2019}, the line spectral estimation is investigated by the Bayesian variational inference using multiple measurement vector (MMV), which outperforms the state-of-the-art MMV methods. Additionally, in~\cite{zhu_multi-snapshot_2019}, a multi-snapshot Newtonized orthogonal matching pursuit (MNOMP) algorithm is given for MMV scenario with relatively low computational complexity. With the prior knowledge of the signal structure, a general SDP method is proposed in~\cite{7145484} to  recover the signal using the  positive trigonometric polynomials, and the perfect signal reconstruction is achieved with sufficient prior information.

The super-resolution methods based on the sparse theory and avoiding the discretization have been proposed. In~\cite{doi:10.1002/cpa.21455},  \emph{total variation norm} is introduced, and show that the exact locations and amplitudes of the line spectrum can be recovered by solving a convex optimization problem. Therefore, the DOA estimation problem can also be described as a type of \emph{line spectral estimation} problem~\cite{6560426}, and a generalized method is proposed in~\cite{8051116} by formulating the sparse signal recovery problems over a continuously indexed dictionary. Then, the \emph{atomic norm} as a specific  form of total variation norm is formulated~\cite{DECASTRO2012336,unser2019,chi_harnessing_2019,7435246}, and an upper bound on the optimization of an atomic norm is given in~\cite{7917313}. Atomic norm minimization (ANM) method~\cite{chi_harnessing_2019} with multiple measurement vectors (MMV)  is proposed~\cite{ZaiYang2016}, and  a Toeplitz covariance matrix reconstruction approach is also given in~\cite{7903732} to formulate a low-rank matrix reconstruction during the DOA estimation. For the general antenna geometries, a method based on total variation minimization is proposed in~\cite{barzegar2017estimation}, where the theoretic guarantee for DOA estimation is derived. In~\cite{8320855}, a family of nonconvex penalties is used to approximate the rank norm, and an iterative reweighted strategy is also proposed to achieve a better performance than the atomic norm method.

However, the existing ANM methods assume the perfect antenna array during the DOA estimation without considering the inconsistent antennas, where a polynomial with steering vector is formulated to estimate the DOA and will be mismatch in the scenario with inconsistent array~\cite{doi:10.1002/cpa.21455}.  The quantized noisy magnitudes are used to reconstruct the sparse signal in~\cite{8713584}, and an  approximation is used for the problem of sparse signal reconstruction for  the approximate message passing method. The Cram\'{e}r-Rao lower bound (CRLB) with quantization is given in~\cite{8340173}, and the algorithm using atomic norm soft thresholding is shown for the sparse reconstruction.  In the practical antenna array, the gain-phase errors among antennas degrade the DOA estimation performance~\cite{5986746,7885556}.
The CS-based method for the DOA estimation is proposed in~\cite{8527362}, and ref.~\cite{8419297} describes the localization method for the near-field sources with gain-phase errors. However, the DOA estimation based on the  \emph{gridless} sparse theory in the scenario has not been proposed.  

In this paper, the DOA estimation problem in the scenario with gain-phase errors has been investigated. The technical contributions of this paper are summarized below:
\begin{itemize}
    \item \textbf{A new atomic norm for DOA estimation with gain-phase errors:} By introducing additional parameters in MMV, a new atomic norm is formulated, and the corresponding dual norm is theoretically obtained.  An optimization problem is formulated for the DOA estimation.
    \item \textbf{An semidefinite program (SDP) problem for the new atomic norm:} To solve the new atomic norm efficiently, an SDP problem is formulated by the Schur complement.
    \item \textbf{Theoretical expressions for the regularization parameter:} In the atomic norm-based method, the regularization parameter determines the DOA estimation performance and is theoretically obtained to describe the reconstruction bound.
\end{itemize}

The remainder of this paper is organized as follows. The DOA estimation model in the uniform linear array (ULA) with gain-phase is formulated in Section~\ref{sec2}. The atomic norm-based DOA estimation method is proposed in Section~\ref{sec3}.  The regularization parameter is theoretically obtained in Section~\ref{sec4}. The CRLB of DOA estimation is given in Section~\ref{sec5}, and the simulation results are shown in Section~\ref{sec6}. Finally, Section~\ref{sec7} concludes the paper.

\textit{Notations:}  $\operatorname{diag}\{\boldsymbol{a}\}$ denotes a diagonal matrix and the diagonal entries are from the vector $\boldsymbol{a}$. $(\cdot)^\text{T}$ and $(\cdot)^\text{H}$ denote  the  matrix transpose and the  Hermitian transpose, respectively. $\|\cdot\|_1$, $\|\cdot\|_2$, $\|\cdot\|_F$ denote  the $\ell_1$ norm, the $\ell_2$ norm, and the Frobenius norm, respectively. $\|\cdot\|^*$ denotes the dual norm. $\boldsymbol{I}_N$ denotes an $N\times N$ identity matrix.   $\otimes$ denotes  the Kronecker product. $\operatorname{Tr}\left\{\cdot\right\}$ denotes the trace of a matrix.  $\mathcal{R}\{a\}$ denotes the real part of complex value $a$.  The boldface capital letters denote the matrix, such as $\boldsymbol{A}$, and  the lower-case letters denote the vector, such as $\boldsymbol{a}$.

\section{System Model With Gain-Phase Errors}\label{sec2}
In an ULA system, the DOA is estimated from the received signal by the antenna array, where a steering vector is used to describe the gain and phase among the perfect antennas. However, the gain-phase errors could cause the model mismatch in characterizing  the steering vector, which eventually degrades the DOA estimation performance. Suppose the DOA estimation problem for $K$ signals in the ULA with unknown gain-phase errors, and the received signal with $P$ snapshots (multiple measurements) can be expressed as 
\begin{align}\label{sys}
	\boldsymbol{Y}=\boldsymbol{GAS}+\boldsymbol{N},
\end{align}
where $\boldsymbol{Y}\in\mathbb{C}^{N\times P}$ and $N$ denotes the number of antennas,  and the spacing between neighboring antennas is $d$. The signals are denoted by a matrix $\boldsymbol{S}\triangleq\begin{bmatrix}
	\boldsymbol{s}_0,\boldsymbol{s}_1,\dots,\boldsymbol{s}_{K-1}
\end{bmatrix}^{\text{T}}$, where the $k$-th signal is defined as $\boldsymbol{s}_k\triangleq\begin{bmatrix}
	s_{k,0},s_{k,1},\dots,s_{k,P-1}
\end{bmatrix}^{\text{T}}$. The steering matrix is denoted as $\boldsymbol{A}\triangleq\begin{bmatrix}
	\boldsymbol{a}(\theta_0),\boldsymbol{a}(\theta_1),\dots,\boldsymbol{a}(\theta_{K-1})
\end{bmatrix}$, where $\theta_k$ is the DOA of the $k$-th signal. The steering vector is defined as 
\begin{align}
	\boldsymbol{a}(\theta)\triangleq\begin{bmatrix}
		1,e^{j \xi\sin\theta},\dots,e^{j (N-1)\xi \sin\theta}
	\end{bmatrix}^{\text{T}},
\end{align}
where $\xi \triangleq \frac{2\pi d}{\lambda}$ and $\lambda$ denotes the wavelength. In the imperfect ULA systems, the received signals are effected by the antenna inconsistency, and we use a diagonal matrix $\boldsymbol{G}$ in (\ref{sys}) to describe the gain-phase errors. The diagonal matrix $\boldsymbol{G}\in\mathbb{C}^{N\times N}$  can be expressed as
\begin{align}\label{G}
	\boldsymbol{G}&\triangleq \left(\boldsymbol{I}_N+\operatorname{diag}\{\boldsymbol{g}\}\right)\operatorname{diag}\{e^{j \boldsymbol{\phi}}\},
\end{align}
where we define $ \boldsymbol{g} \triangleq\begin{bmatrix}
g_0,g_1,\dots,g_{N-1}
\end{bmatrix}^{\text{T}}$ as the gain-error vector ($g_n\in\mathbb{R} $)  and $\boldsymbol{\phi} \triangleq\begin{bmatrix}
\phi_0,\phi_1,\dots,\phi_{N-1}
\end{bmatrix}^{\text{T}}$ as the phase-error vector ($\phi_n \in[0,2\pi)$).

In this paper,  by exploiting the signal sparsity in the spatial domain, we will estimate the DOA parameters $\boldsymbol{\theta}=\begin{bmatrix}
	\theta_0,\theta_1,\dots,\theta_{K-1}
\end{bmatrix}^{\text{T}}$ from the received signal $\boldsymbol{Y}$ with the unknown antenna inconsistency including the gain errors $\boldsymbol{g}$ and the phase errors $\boldsymbol{\phi}$. To avoid the discretized grids in the spatial domain, we will propose a new atomic norm and formulate the DOA estimation problem as an optimization problem with new atomic norm.

\section{Atomic Norm-Based Gridless DOA Estimation}\label{sec3}
\subsection{Preliminary Atomic Norm}
To improve the DOA estimation performance by exploiting the signal sparsity, the ANM-based methods have been proposed. Different from the exiting sparse-based methods using a dictionary matrix formulated by the discretized angles, such as the $\ell_1$ norm method~\cite{6129390,8667328,5783354}, the mixed $\ell_{2,0}$ norm approximation~\cite{5466152}, the ANM methods reconstruct the sparse signals without discretizing the spatial domain, so the ANM methods are the \emph{gridless} sparse methods~\cite{6998075}.

Usually, for the DOA estimation with perfect antennas, the  system mode is formulated  as
\begin{align}\label{perfect}
	\boldsymbol{Y}=\boldsymbol{AS}+\boldsymbol{N},
\end{align}
so the \emph{atomic set} is defined as~\cite{8259341,7484756,7313018}
\begin{align}
	\mathcal{A}\triangleq \left\{\boldsymbol{a}(\theta)\boldsymbol{b}^{\text{T}}: \theta\in[0,2\pi), \|\boldsymbol{b}\|_2=1\right\}.
\end{align}
Then, the DOA estimation in (\ref{perfect}) is transferred into the following optimization problem (ANM)
\begin{align} 
	\min_{\boldsymbol{X}} \frac{1}{2} \|\boldsymbol{Y}-\boldsymbol{X}\|^2_{F}+\tau\|\boldsymbol{X}\|_{\mathcal{A}}, 
\end{align} 
where the \emph{atomic norm} is defined as
\begin{align}
	\|\boldsymbol{X}\|_{\mathcal{A}}&=\inf\left\{t>0: \boldsymbol{X}\in t\operatorname{conv}(\mathcal{A})\right\}\\
	& = \inf\left\{
	\sum_k c_k:\boldsymbol{X}=\sum_k c_k \boldsymbol{a}(\theta)\boldsymbol{b}^{\text{T}}, c_k\geq 0
	\right\}.\notag
\end{align}
Then, the ANM problem can be solved by SDP~\cite{6560426,DBLP05238,6576276}.

\subsection{New Atomic Norm Method for DOA Estimation}
\subsubsection{The Definition of New Atomic Norm}
Without the noise, the received signal can be also expressed as
\begin{align}\label{X}
\boldsymbol{X}=\boldsymbol{GAS}= \sum^{K-1}_{k=0}(\boldsymbol{I}_N+\operatorname{diag}\{\boldsymbol{e}\})\boldsymbol{a}(\theta_k)\boldsymbol{s}^{\text{T}}_k.
\end{align}
From the gain-phase model in (\ref{G}), the antenna inconsistency $\boldsymbol{e}$  in (\ref{X}) is $\boldsymbol{e}=\left(\boldsymbol{1}_N+  \boldsymbol{g}  \right)\operatorname{diag}\{e^{j \boldsymbol{\phi}}\}-\boldsymbol{1}_N$, where $\boldsymbol{1}_N$ is a $N\times 1$ vector with all the entries being $1$.

Then, we propose a new atomic decomposition to describe $\boldsymbol{X}$ in the scenario with gain-phase errors and  to improve the  robustness, and it is defined as  
\begin{align}\label{A0}
\|\boldsymbol{X}\|_{\tilde{\mathcal{A}},0}\triangleq \Bigg\{
K: &\boldsymbol{X}=\sum_{k=0}^{K-1} b_k (\operatorname{diag}\{\boldsymbol{e}\}+\boldsymbol{I}_N) \boldsymbol{a}(\theta_k)\boldsymbol{d}_k^{\text{T}},\notag \\
&\|\boldsymbol{e}\|_2 \leq C_{\text{e}}, \|\boldsymbol{d}_k\|_2 \leq 1,b_k\geq 0 
\Bigg\},
\end{align}
where $C_{\text{e}}$ is used to control the gain and phase errors.
Note that, the $\ell_2$ norm for the gain-phase error $\| \boldsymbol{e}\|_2\leq C_{\text{e}}$ can be easily extended to the sparse norm $\| \boldsymbol{e}\|_1\leq C_{\text{e}}$. When we have $\| \boldsymbol{e}\|_1\leq C_{\text{e}}$, we can obtain $\| \boldsymbol{e}\|_2\leq C_{\text{e}}$ with $\|\boldsymbol{e}\|_1\geq \|\boldsymbol{e}\|_2$.  Therefore, the proposed atomic norm can be used in the scenario with the sparse gain-phase errors, where only a few antennas are inconsistent.

However, it is not computationally feasible  to find the minimum $K$ in (\ref{A0}) by the atomic decomposition of $\boldsymbol{X}$. A new atomic norm $\ell_{\tilde{\mathcal{A}}}$ is proposed by a convex relaxation of $\ell_{\tilde{\mathcal{A}},0}$, and is defined as
\begin{align}\label{atomic}
	\|\boldsymbol{X}\|_{\tilde{\mathcal{A}}}&\triangleq  \inf\Bigg\{
 \sum_k b_k \bigg|  
	\boldsymbol{X}=\sum_{k} b_k(\operatorname{diag}\{\boldsymbol{e}\}+\boldsymbol{I}_N) \boldsymbol{a}(\theta_k)\boldsymbol{d}_k^{\text{T}},\notag\\
 	&\qquad\qquad \|\boldsymbol{e}\|_2 \leq C_{\text{e}}, \|\boldsymbol{d}_k\|_2 \leq 1,b_k\geq 0  
	\Bigg\}\\
	& = \inf\Bigg\{
 \|\boldsymbol{b}\|_1 \bigg|
 \boldsymbol{X}=\sum_{k} b_k(\operatorname{diag}\{\boldsymbol{e}\}+\boldsymbol{I}_N) \boldsymbol{a}(\theta_k)\boldsymbol{d}_k^{\text{T}},\notag\\
 &\qquad\qquad \|\boldsymbol{e}\|_2 \leq C_{\text{e}}, \|\boldsymbol{d}_k\|_2 \leq 1,b_k\geq 0 
	\Bigg\},\notag
\end{align}
where $\boldsymbol{b}$ is defined as $\boldsymbol{b}\triangleq \begin{bmatrix}
b_0,b_1,\dots,b_{K-1}
\end{bmatrix}^{\text{T}}$. This optimization is named as Gain-Phase ANM (GP-ANM) to be different  from the existing ANM methods. From GP-ANM, the novel DOA estimation method will be proposed, and the corresponding algorithm will be given in the following sections. 

\subsubsection{DOA Estimation Using GP-ANM}

With the received signals $\boldsymbol{Y}$ and the additive noise $\boldsymbol{N}$, the DOA estimation problem can be described by following optimization problem
\begin{align}\label{op}
	\min_{\boldsymbol{X}} \frac{1}{2} \|\boldsymbol{Y}-\boldsymbol{X}\|^2_{F}+\tau\|\boldsymbol{X}\|_{\tilde{\mathcal{A}}}, 
\end{align}
where the first term is used to control the reconstruction performance and the second one is for the sparsity of $\boldsymbol{X}$. The regularization parameter $\tau$ is adopted to control the balance between the reconstruction performance and the sparsity. We will show how to get the regularization parameter $\tau$ in the following sections. The optimization problem in (\ref{op}) is a special case in~\cite{16M1071730}, where a generalization of SDP over infinite dictionary is investigated.

Before solving the optimization problem (\ref{op}), we first introduce the dual norm~\cite{7313018} for the proposed atomic norm. We define the dual norm of atomic norm as
\begin{align}
\|\boldsymbol{U}\|^*_{\tilde{\mathcal{A}}}\triangleq \sup_{\|\boldsymbol{X}\|_{\tilde{\mathcal{A}}}\leq 1} \langle \boldsymbol{X},\boldsymbol{U}\rangle, 
\end{align}
where atomic norm is given in (\ref{atomic}). 

Based on the dual norm,  the dual problem of (\ref{op}) can be obtained from the following proposition
\begin{theorem}
For an optimization problem $\min_{\boldsymbol{X}} \frac{1}{2} \|\boldsymbol{Y}-\boldsymbol{X}\|^2_{F}+\tau\|\boldsymbol{X}\|_{\tilde{\mathcal{A}}}$, where $\boldsymbol{Y} \in\mathbb{C}^{N\times P}$, $\boldsymbol{X}\in \mathbb{C}^{N\times P}$ and $\tau\geq 0$, the dual problem is 
\begin{align}\label{dual13}
\min_{\boldsymbol{U}} \quad & \|\boldsymbol{Y}-\boldsymbol{U}\|^2_{F}\\
\text{s.t.}\quad & \|\boldsymbol{U}\|^*_{\tilde{\mathcal{A}}}\leq \tau,  \notag
\end{align}
where $\|\boldsymbol{U}\|^*_{\tilde{\mathcal{A}}}$ denotes the dual norm of $\|\boldsymbol{U}\|_{\tilde{\mathcal{A}}}$.
\end{theorem}

\begin{proof}
Using a Lagrange multiplier $\boldsymbol{U}$, we first formulate the following Lagrange function of the optimization problem (\ref{op}) as
\begin{align}
L(\boldsymbol{X},\boldsymbol{Z},\boldsymbol{U})\triangleq \frac{1}{2}\|\boldsymbol{Y}-\boldsymbol{Z}\|^2_F+\tau\|\boldsymbol{X}\|_{\tilde{\mathcal{A}}}+\langle \boldsymbol{Z}-\boldsymbol{X},\boldsymbol{U}\rangle, 
\end{align}
where the inner product between matrices is defined as $\langle \boldsymbol{X},\boldsymbol{Y} \rangle\triangleq \mathcal{R}\{\operatorname{Tr}(\boldsymbol{Y}^{\text{H}}\boldsymbol{X})\}$. Using the Lagrange function, the dual problem of (\ref{op}) is given as~\cite{boyd2004convex} 
\begin{align}\label{dual}
\max_{\boldsymbol{U}}\min_{\boldsymbol{X},\boldsymbol{Z}} L(\boldsymbol{X},\boldsymbol{Z},\boldsymbol{U}) = \max_{\boldsymbol{U}}\left\{L_1(\boldsymbol{U}) -L_2(\boldsymbol{U})\right\},
\end{align}
where $L_1(\boldsymbol{U})\triangleq \min_{\boldsymbol{Z}} \frac{1}{2}\|\boldsymbol{Y}-\boldsymbol{Z}\|^2_F+\langle \boldsymbol{Z},\boldsymbol{U}\rangle$, and $L_2(\boldsymbol{U})\triangleq \max_{\boldsymbol{X}}\left\{\langle \boldsymbol{X},\boldsymbol{U}\rangle-\tau\|\boldsymbol{X}\|_{\tilde{\mathcal{A}}}\right\}$.

Then, with the definition of dual norm, $L_2(\boldsymbol{U})$ can be simplified as 
\begin{align}
L_2(\boldsymbol{U}) 
& = \tau \max_{\boldsymbol{X}}\left\{\left\langle \boldsymbol{X},\frac{1}{\tau}\boldsymbol{U}\right\rangle-\|\boldsymbol{X}\|_{\tilde{\mathcal{A}}}\right\}\notag\\
& = \tau I\left(\left\|\boldsymbol{U}\right\|^*_{\tilde{\mathcal{A}}}\leq \tau\right),
\end{align}
where the indicate function is defined as
\begin{align}
I\left(\left\|\boldsymbol{U}\right\|^*_{\tilde{\mathcal{A}}}\leq \tau\right)=\begin{cases}
0, & \left\|\boldsymbol{U}\right\|^*_{\tilde{\mathcal{A}}}\leq \tau\\
\infty,&\text{otherwise}
\end{cases}.
\end{align}
Additionally, $L_1(\boldsymbol{U})$ can be obtained as
\begin{align}
L_1(\boldsymbol{U}) = -\frac{1}{2}\|\boldsymbol{Y}-\boldsymbol{U}\|_F^2+\frac{1}{2}\|\boldsymbol{Y}\|^2_F.
\end{align}

Therefore, the dual problem in (\ref{dual}) can be simplified as (\ref{dual13}).	
\end{proof}

With the dual problem, we obtain a SDP problem to solve  (\ref{dual}) efficiently. The SDP problem is given as the following proposition.
\begin{theorem}\label{prop2} 	
	With the gain-phase error vector $\boldsymbol{e}$, the feasible set of the dual problem (\ref{dual13}) for DOA estimation with multiple measurements 
	is included by the following SDP problem, so  the dual problem (\ref{dual13})	can be  relaxed and simplified as 
	\begin{align}\label{sdp}
	\min_{\substack{
		\boldsymbol{U}\in\mathbb{C}^{N\times P}\\ \boldsymbol{Q}\in\mathbb{C}^{N\times N}}} \quad & \|\boldsymbol{Y}-\boldsymbol{U}\|_{F}\\
	\text{s.t.}\quad & \begin{bmatrix}
	\boldsymbol{Q} & \boldsymbol{U}\\
	\boldsymbol{U}^{\text{H}}& \tau^2 \boldsymbol{I}_P
	\end{bmatrix}\succeq 0\notag \\
	& \sum_n Q_{n,n+k}= 0\ (k\neq 0) \notag\\
		& \operatorname{Tr}(\boldsymbol{Q}) +(C_{\text{e}}+2\sqrt{N})C_{\text{e}} \left\|\boldsymbol{Q}\right\|_2 -1\leq 0  \notag\\
	&\boldsymbol{Q}\text{ is Hermitian},\notag  
	\end{align}  
	where the $\ell_2$ norm of a matrix $\|\boldsymbol{Q}\|_2$ is the largest singular value of $\boldsymbol{Q}$.
\end{theorem}

To show that the optimization problem (\ref{sdp}) is a type of SDP problem, we give the proof in Appendix~\ref{sdpAP}. The proof for Proposition~\ref{prop2} is given in Appendix~\ref{ap2}. Note that the constraints in the SDP problem (\ref{sdp}) are  sufficient to the dual norm constraint in (\ref{dual13}),  so the denoised result $\boldsymbol{U}$ in (\ref{sdp}) is only a sufficient approximation of the optimal results in (\ref{dual13}).

To estimate the DOA from the solution of (\ref{sdp}), we can formulate a quadratic $\|\hat{\boldsymbol{U}}\boldsymbol{a}(\theta)\|^2_2$, which is also a polynomial of $\boldsymbol{a}(\theta)$. Inspired by~\cite{chi_harnessing_2019}, we can estimate the DOA by searching the peak of $\|\hat{\boldsymbol{U}}\boldsymbol{a}(\theta)\|^2_2$, so we get a straightforward corollary.
\begin{cor}	
	The DOA polynomial is formulated as
	\begin{align}
	\left\|\hat{\boldsymbol{U}}\boldsymbol{a}(\theta)\right\|^2_2\leq \tau^2\left(1-(C_{\text{e}}+2\sqrt{N})C_{\text{e}} \|\hat{\boldsymbol{Q}}\|_2\right),
	\end{align}
	where $\hat{\boldsymbol{U}}$ and $\hat{\boldsymbol{Q}}$ are the solutions of the SDP problem in (\ref{sdp}). Additionally, the quantity $1-(C_{\text{e}}+2\sqrt{N})C_{\text{e}} \|\hat{\boldsymbol{Q}}\|_2$ is positive.
\end{cor}
\begin{proof}
	From (\ref{schur}), we can find that for any $\theta$, we have
	\begin{align}
	\|\hat{\boldsymbol{U}}\boldsymbol{a}(\theta)\|_2^2\leq \tau^2 \boldsymbol{a}^{\text{H}}(\theta)\hat{\boldsymbol{Q}}\boldsymbol{a}(\theta).
	\end{align}
	Therefore, from the construction of $\hat{\boldsymbol{Q}}$ in (\ref{Q}), we obtain
		\begin{align}
	\|\hat{\boldsymbol{U}}\boldsymbol{a}(\theta)\|_2^2\leq \tau^2\left(1-(C_{\text{e}}+2\sqrt{N})C_{\text{e}} \|\hat{\boldsymbol{Q}}\|_2\right).
	\end{align} 
\end{proof}
Then, the DOA can be obtained by searching the peak values of $ 	\left\|\hat{\boldsymbol{U}}\boldsymbol{a}(\theta)\right\|^2_2$, which is closed to  $\tau^2\left(1-(C_{\text{e}}+2\sqrt{N})C_{\text{e}} \|\hat{\boldsymbol{Q}}\|_2\right)$. The estimated DOAs are denoted as $\hat{\theta}_k$ $(k=0,1,\dots,K-1)$.

To estimate the other unknown parameters including $\boldsymbol{e}$, $b_k$, $\boldsymbol{d}_k$, we formulate the following optimization problem
\begin{align}\label{ls}
\min_{\boldsymbol{e},\boldsymbol{b}_k,\boldsymbol{d}_k}\quad & \|\boldsymbol{Y}-\sum_{k} b_k(\operatorname{diag}\{\boldsymbol{e}\}+\boldsymbol{I}) \boldsymbol{a}(\hat{\theta}_k)\boldsymbol{d}_k^{\text{T}}\|_F\\
\text{s.t.}\quad &
\|\boldsymbol{e}\|_2 = C_{\text{e}}, \|\boldsymbol{d}_k\|_2 \leq 1,b_k\geq 0.  \notag
\end{align}
Since the upper bound of the gain-phase errors $\boldsymbol{e}$ is used in the proposed GP-ANM method to estimate the DOA, the constraint $\|\boldsymbol{e}\|_2 = C_{\text{e}}$ is used for the estimation of unknown parameters. For the gain-phase errors $\boldsymbol{e}$, we can formulate the following optimization problem 
\begin{align}
	\min_{\boldsymbol{e}}\quad & f(\boldsymbol{e})\triangleq \|\boldsymbol{Y}-\sum_{k} b_k(\operatorname{diag}\{\boldsymbol{e}\}+\boldsymbol{I}) \boldsymbol{a}(\hat{\theta}_k)\boldsymbol{d}_k^{\text{T}}\|^2_F\\
\text{s.t.}\quad &
\|\boldsymbol{e}\|_2 = C_{\text{e}}.  \notag
\end{align}
$f(\boldsymbol{e})$ can be rewritten as
\begin{align}
	f(\boldsymbol{e}) & = \|\boldsymbol{Y}-(\operatorname{diag}\{\boldsymbol{e}\}+\boldsymbol{I}) \underbrace{\sum_{k} b_k \boldsymbol{a}(\hat{\theta}_k)\boldsymbol{d}_k^{\text{T}}}_{\boldsymbol{H}}\|_F^2\\
	& =  \sum_{n=0}^{N-1}\left\|\bar{\boldsymbol{y}}_{n}-(e_n+1)\bar{\boldsymbol{h}}_n
	\right\|^2_2,\notag
\end{align}
where $\bar{\boldsymbol{y}}_n$ denotes the $n$-th row of $\boldsymbol{Y}$ and $\bar{\boldsymbol{h}}_{n}$ is the $n$-th row of $\boldsymbol{H}$.
Therefore, the Lagrange function for $\boldsymbol{e}$ with the Lagrange parameter $\lambda_{\text{e}}\geq 0$ can be obtained as 
\begin{align}
	L(\boldsymbol{e}) =f(\boldsymbol{e})+\lambda_{\text{e}}(\|\boldsymbol{e}\|^2_2-C^2_{\text{e}}).
\end{align}
With $\frac{\partial L(\boldsymbol{e}) }{\partial \boldsymbol{e}^*} = 0$, we can obtain the estimated gain-phase as
\begin{align}\label{eq27}
	\hat{e}_n = \frac{\bar{\boldsymbol{h}}_n^{\text{H}}(\bar{\boldsymbol{y}}_n-\bar{\boldsymbol{h}}_n)}{\lambda_{\text{e}}+\bar{\boldsymbol{h}}_n^{\text{H}}\bar{\boldsymbol{h}}_n},
\end{align}
where we choose $\lambda_{\text{e}}$ to ensure that $\|\boldsymbol{e}\|_2=C_{\text{e}}$.

For the unknown parameter $b_k$ and $\boldsymbol{d}_k$, we can formulate $\boldsymbol{d}'_k=\boldsymbol{d}_kb_k$. With the estimated $\hat{\theta}_k$ and $\hat{\boldsymbol{e}}$, we have
\begin{align}
\min_{\boldsymbol{d}'_k}\quad & f(\boldsymbol{d}'_k)\triangleq \|\boldsymbol{Y}-\sum_{k} (\operatorname{diag}\{\hat{\boldsymbol{e}}\}+\boldsymbol{I}) \boldsymbol{a}(\hat{\theta}_k)\boldsymbol{d}_k^{'\text{T}}\|^2_F.
\end{align}
Then, we define $\boldsymbol{D}'=[\boldsymbol{d}'_0,\boldsymbol{d}'_1,\dots, \boldsymbol{d}'_{K-1}]$ and $\hat{\boldsymbol{A}}=[\boldsymbol{a}(\hat{\theta}_0),\boldsymbol{a}(\hat{\theta}_1),\dots,\boldsymbol{a}(\hat{\theta}_{K-1})]$, and $\boldsymbol{D}'$ can be estimated as
\begin{align}\label{eq29}
\hat{\boldsymbol{D}}' = \left[\hat{\boldsymbol{A}}^{\dagger}(\operatorname{diag}\{\hat{\boldsymbol{e}}\}+\boldsymbol{I})^{-1}\boldsymbol{Y}\right]^{\text{T}},
\end{align}
where $\dagger$ denotes the pseudo-inverse operation.
Then,  $\boldsymbol{d}_k$ can be estimated from the normalized $\boldsymbol{d}'_k$ and $b_k$ is the normalization coefficient. By alternatively estimating the unknown parameters$\boldsymbol{e}$, $b_k$, $\boldsymbol{d}_k$. We can finally estimated all the unknown parameters.

The details of the proposed method for the DOA estimation is given in Algorithm~\ref{alg1}. The computational complexity of the proposed method is almost the same with the traditional atomic norm minimization (ANM) method. Only a $\ell_2$ norm for the matrix $\boldsymbol{Q}$ is added in the SDP problem, so the computational complexity is $\mathcal{O}(N^3)$ more than the ANM method at each iteration.

\begin{algorithm}%[t]
	\caption{DOA Estimation Using GP-ANM} \label{alg1}
	\begin{algorithmic}[1]
		\STATE  \emph{Input:} received signal $\boldsymbol{Y}$, noise variance $\sigma^2_{\text{n}}$, the number of antennas $N$, and the number of measurements (snapshots) $P$.
		\STATE \emph{Initialization:} $\tau =\eta \sigma_{\text{n}}\sqrt{4NP\ln (N)}$.
		 \STATE Formulate the SDP problem as (\ref{sdp}), and obtain the matrix $\hat{\boldsymbol{U}}$.
		 \STATE Get the polynomial $f(\hat{\boldsymbol{U}})=\left\|\hat{\boldsymbol{U}}\boldsymbol{a}(\theta)\right\|^2_2$.
		 \STATE Use the peak searching of $f(\hat{\boldsymbol{U}})$,  and get the estimated DOA $\hat{\boldsymbol{\theta}}$.
		 \STATE The other unknown parameters can be obtained by the alternative estimations in the problem (\ref{ls}).
		\STATE \emph{Output:}  the estimated DOA $\hat{\boldsymbol{\theta}}$.
	\end{algorithmic}
\end{algorithm}

\section{The regularization parameter $\tau$}\label{sec4}
In (\ref{op}), the regularization parameter is important and has a great effect on reconstructing the sparse signal, so we will obtain the  regularization parameter in this section. 

Usually, the regularization parameter $\tau$ can be chosen as~\cite{DBLP:journals/corr/LiYTW17} 
\begin{align}
\tau\approx \eta \mathcal{E}\left\{\|\boldsymbol{N}\|^*_{\tilde{\mathcal{A}}}\right\}\quad (\eta \geq 1).
\end{align}
To get $\mathcal{E}\left\{\|\boldsymbol{N}\|^*_{\tilde{\mathcal{A}}}\right\}$, we can obtain the following proposition to determine the regularization
\begin{theorem}\label{prop4}
The entries in $\boldsymbol{N}\in\mathbb{C}^{N\times P}$ follow the  zero-mean Gaussian distribution with the variance being $\sigma_{\text{N}}^2$ and the entries are independent. With the probability more than $1-2e^{-t^2/2}$, the upper bound of $\left\{\|\boldsymbol{N}\|^*_{\tilde{\mathcal{A}}}\right\}$ can be obtained as
\begin{align}
\mathcal{E}\left\{\|\boldsymbol{N}\|^*_{\mathcal{A}}\right\}\leq \min\left\{\text{bd}_1,\text{bd}_2\right\}C_{\text{e}} +
\sigma_{\text{N}}\sqrt{4NP\ln N},
\end{align}
where  the definition of dual atomic norm is defined in (\ref{17}), $ \text{bd}_1\triangleq \sqrt{2}\sigma_{\text{N}}\frac{\Gamma((NP+1)/2)}{\Gamma(NP/2)}$, and $\text{bd}_2\triangleq \left(\sqrt{N}+\sqrt{P}+t\right)\sigma_{\text{N}}$.
\end{theorem}
The proof for Proposition~\ref{prop4} is given in Appendix~\ref{ap3}. Then, the regularization parameter $\tau$ is 
\begin{align}
\tau\approx \eta \left(\min\left\{\text{bd}_1,\text{bd}_2\right\}C_{\text{e}} +
\sigma_{\text{N}}\sqrt{4NP\ln N}\right).
\end{align} 
We will show that with the regularization parameter $\tau$, the probability of $\|\boldsymbol{N}\|^*_{\tilde{\mathcal{A}}}\geq \frac{\tau}{\eta}$ can be obtained as
\begin{align}
& \mathbb{P}\left(\|\boldsymbol{N}\|^*_{\tilde{\mathcal{A}}}\geq \frac{\tau}{\eta} \right)\\
%&= \mathbb{P}\left(\sup_{\substack{ 
%		\|\boldsymbol{e}\|_2\leq C_{\text{e}} 
%		\\
%		\theta\in[0,2\pi) }} \left\|  \boldsymbol{N}^{\text{H}}  \left[\boldsymbol{e} 
%+ \boldsymbol{a}(\theta) \right]
%\right\|_2 \geq\frac{\tau}{\eta}\right)\notag\\
 &= 
\mathbb{P}\left(\sup_{\substack{ 
		\|\boldsymbol{e}\|_2\leq C_{\text{e}} 
		\\
		\theta\in[0,2\pi) }} \left\|  \boldsymbol{N}^{\text{H}}  \left[\boldsymbol{e} 
+ \boldsymbol{a}(\theta) \right]
\right\|_2 \geq  \alpha+
\beta\right)  \notag\\
 & \leq 
\mathbb{P}\left(\sup_{	\|\boldsymbol{e}\|_2\leq C_{\text{e}} } \left\|  \boldsymbol{N}^{\text{H}}  \boldsymbol{e}  
\right\|_2  + \sup_{  \theta\in[0,2\pi)  } \left\|  \boldsymbol{N}^{\text{H}}   
\boldsymbol{a}(\theta) \right\|_2 \geq \alpha+
\beta\right)  \notag\\
% &=
%\mathbb{P}\left( C_{\text{e}}  \left\|  \boldsymbol{N} 
%\right\|_2  + \sup_{  \theta\in[0,2\pi)  } \left\|  \boldsymbol{N}^{\text{H}}   
%\boldsymbol{a}(\theta) \right\|_2 \geq  \alpha+
%\beta\right)  \notag\\
 &\leq 
\mathbb{P}\left( C_{\text{e}}  \left\|  \boldsymbol{N} 
\right\|_F  + \sup_{  \theta\in[0,2\pi)  } \left\|  \boldsymbol{N}^{\text{H}}   
\boldsymbol{a}(\theta) \right\|_2 \geq  \alpha+
\beta\right)  \notag\\
% &\leq 
%\mathbb{P}\left(  C_{\text{e}} \left\|  \boldsymbol{N} 
%\right\|_F   \geq  \alpha \right)  +
%\mathbb{P}\left( C_{\text{e}}  \left\|  \boldsymbol{N} 
%\right\|_F \leq  \alpha \right) \notag\\
%&\quad
%\mathbb{P}\left( \sup_{  \theta\in[0,2\pi)  } \left\|  \boldsymbol{N}^{\text{H}}   
%\boldsymbol{a}(\theta) \right\|_2 \geq 
%\beta\right) 
%\notag\\
%& \leq 2e^{-t^2/2}+\left(1-2e^{-t^2/2}\right)\frac{1}{N^2}\notag\\
& \leq \frac{1}{N^2}+\left(1- \frac{1}{N^2}\right)z(N,P), \notag
\end{align} 
where we define $\alpha\triangleq \min\left\{\text{bd}_1,\text{bd}_2\right\}C_{\text{e}} $, 
%$\alpha\triangleq \left(\sqrt{N}+\sqrt{P}+t\right)\sigma_{\text{N}}$,  
$\beta\triangleq \sigma_{\text{N}}\sqrt{4NP\ln N}$ and
\begin{align}
z(N,P)=\begin{cases}2e^{-t^2/2}, \sqrt{N}+\sqrt{P}+t \leq \sqrt{2}\frac{\Gamma((NP+1)/2)}{\Gamma(NP/2)}\\
\frac{\Gamma\left(NP/2,\left(\sqrt{N}+\sqrt{P}+t\right)^2/2\right)}{\Gamma(NP/2)}, \text{otherwise}
\end{cases}.
\end{align} 
The incomplete Gamma function is defined as $\Gamma(s,x)\triangleq \int^{\infty}_x t^{s-1}e^{-t}\,dt$.

\begin{figure}
	\centering
	\includegraphics[width=3in]{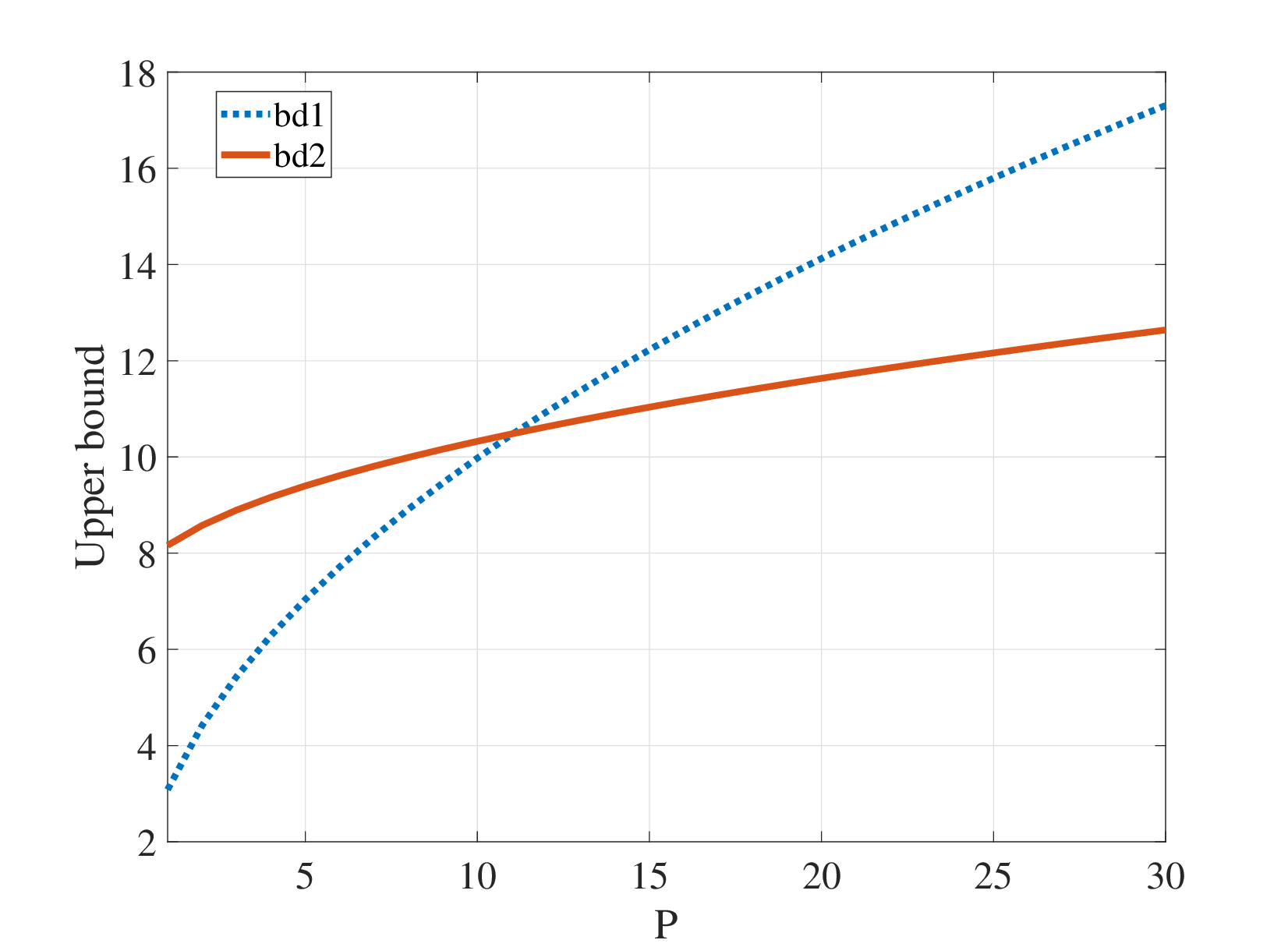}
	\caption{The upper bound of $\mathcal{E}\left\{  \left\|  \boldsymbol{N} \right\|_2\right\}$.}
	\label{upper}
\end{figure}

When we can choose $t=4$, the probability of $\mathcal{E}\left\{  \left\|  \boldsymbol{N} \right\|_2\right\} \leq \sqrt{N}+\sqrt{P}+4$ is more than $0.9993$. In Fig.~\ref{upper}, we show the two types of upper bounds, where the antenna number $N$ is $10$, and the number of measurements $P$ is from $1$ to $30$. When $P\leq 11$, we have $\text{bd}_1\leq \text{bd}_2$, and  $\text{bd}_1>\text{bd}_2$ with $P> 11$. Hence, for larger $P$, we choose  $\text{bd}_2$ as the tighter upper bound, and for smaller $P$, we choose $\text{bd}_1$.  

Therefore,  according to Theorem III.6 in~\cite{8259341}, with probability $1-\frac{1}{N^2}-\left(1- \frac{1}{N^2}\right)z(N,P)$, the reconstruction error is limited by
\begin{align}\label{eq35}
\left\|\hat{\boldsymbol{X}}-\boldsymbol{X}_{*}\right\|^2_F\leq \tau^2,
\end{align}
where $\hat{\boldsymbol{X}}$ denotes the estimated $\boldsymbol{X}$ by minimizing the atomic norm, and $\boldsymbol{X}_{*}$ denotes the ground-truth $\boldsymbol{X}$.

\section{Sparse Gain-Phase Errors}
In the scenario with only a few gain-phase errors, the gain-phase errors are sparse. The proposed type of atomic norm can be rewritten as 
\begin{align} 
	\|\boldsymbol{X}\|_{\tilde{\mathcal{A}}}& = \inf\Bigg\{
 \|\boldsymbol{b}\|_1 \bigg|
 \boldsymbol{X}=\sum_{k} b_k(\operatorname{diag}\{\boldsymbol{e}\}+\boldsymbol{I}_N) \boldsymbol{a}(\theta_k)\boldsymbol{d}_k^{\text{T}},\notag\\
 &\qquad\qquad \|\boldsymbol{e}\|_1 \leq C_{\text{e}}, \|\boldsymbol{d}_k\|_2 \leq 1,b_k\geq 0 
	\Bigg\},\notag
\end{align}
where the $\ell_1$ norm $\|\boldsymbol{e}\|_1$ is used to describe the sparse gain-phase errors. We formulate $\boldsymbol{e}'\triangleq \operatorname{diag}(\boldsymbol{e})\boldsymbol{a}(\theta) \in\mathbb{C}^{N\times 1}$, then, the $\ell_1$ norm of $\boldsymbol{e}'$ can be simplified as 
	\begin{align}
		\|\boldsymbol{e}'\|_1 =\|\operatorname{diag}(\boldsymbol{e})\boldsymbol{a}(\theta)\|_1  =\|\boldsymbol{e}\|_1\leq C_{\text{e}}.\notag 
	\end{align}
	Therefore,  the corresponding dual norm can be obtained as
\begin{align}
\|\boldsymbol{U}\|^*_{\tilde{\mathcal{A}}} &= \sup_{\|\boldsymbol{X}\|_{\tilde{\mathcal{A}}}\leq 1} \langle \boldsymbol{X},\boldsymbol{U}\rangle\\
& = \sup_{\substack{ 
		\|\boldsymbol{e}\|_1\leq C_{\text{e}}
		\\
		\theta\in[0,2\pi) }}  \left\|  \boldsymbol{U}^{\text{H}}  \left[\boldsymbol{e} 
+ \boldsymbol{a}(\theta) \right]
\right\|_2, 
\notag 
\end{align}
where  we  reuse the notation $\boldsymbol{e}$ instead of $\boldsymbol{e}'$ to avoid introducing additional symbol $\boldsymbol{e}'$.

If $\|\boldsymbol{e}\|_2\leq \frac{1}{\sqrt{N}}C_\text{e}$, we have $\|\boldsymbol{e}\|_1\leq C_{\text{e}}$. Therefore, we can obtain 
\begin{align}
\|\boldsymbol{U}\|^*_{\tilde{\mathcal{A}}} & \leq \sup_{\substack{ 
		\|\boldsymbol{e}\|_2\leq \frac{C_{\text{e}}}{\sqrt{N}} 
		\\
		\theta\in[0,2\pi) }}  \left\|  \boldsymbol{U}^{\text{H}}  \left[\boldsymbol{e} 
+ \boldsymbol{a}(\theta) \right]
\right\|_2\\
  & \leq \sup_{\substack{   
		\theta\in[0,2\pi) }}  \tau^2\left[\left(2+\frac{1}{\sqrt{N}}\right)C_{\text{e}} \|\boldsymbol{Q}\|_2+  \boldsymbol{a}^\text{H}(\theta)\boldsymbol{Q} \boldsymbol{a}(\theta) \right]. \notag
\end{align} 
Similarly, the SDP problem with the $\ell_1$ norm in atomic norm can be obtained as
\begin{align}
	\min_{\substack{
		\boldsymbol{U}\in\mathbb{C}^{N\times P}\\ \boldsymbol{Q}\in\mathbb{C}^{N\times N}}} \quad & \|\boldsymbol{Y}-\boldsymbol{U}\|_{F}\\
	\text{s.t.}\quad & \begin{bmatrix}
	\boldsymbol{Q} & \boldsymbol{U}\\
	\boldsymbol{U}^{\text{H}}& \tau^2 \boldsymbol{I}_P
	\end{bmatrix}\succeq 0\notag \\
	& \sum_n Q_{n,n+k}= 0\ (k\neq 0) \notag\\
		& \operatorname{Tr}(\boldsymbol{Q}) +\left(2+\frac{1}{\sqrt{N}}\right)C_{\text{e}} \left\|\boldsymbol{Q}\right\|_2 -1\leq 0  \notag\\
	&\boldsymbol{Q}\text{ is Hermitian},\notag  
	\end{align}
which can be solved efficiently.

\section{CRLB For DOA Estimation With Gain-Phase Errors}\label{sec5}
For the DOA estimation problem $\boldsymbol{Y}=\boldsymbol{GAS}+\boldsymbol{N}$, we use $\boldsymbol{A}\triangleq \begin{bmatrix}
\boldsymbol{a}(\theta_0),\boldsymbol{a}(\theta_1),\dots,\boldsymbol{a}(\theta_{K-1})
\end{bmatrix}$ to denote the steering matrix, and we assume $\boldsymbol{n}=\operatorname{vec}\{\boldsymbol{N}\}\sim\mathcal{CN}(\boldsymbol{0},\sigma_{\text{n}}^2\boldsymbol{I})$. We consider $K$ unknown signals $\boldsymbol{S}=\begin{bmatrix}
\boldsymbol{s}_0,\boldsymbol{s}_1,\dots,\boldsymbol{s}_{K-1}
\end{bmatrix}^{\text{T}}$, and we assume $\boldsymbol{s} $ follows the zero mean Gaussian distribution with $\mathcal{E}(\boldsymbol{s}\boldsymbol{s}^{\text{H}})=\boldsymbol{B}$ and $\boldsymbol{s}\sim\mathcal{CN}(\boldsymbol{0},\boldsymbol{B})$, where $\boldsymbol{s}\triangleq \operatorname{vec}\{\boldsymbol{S}\}$. Then, in this section, the CRLB  will be derived theoretically to indicate the DOA estimation performance of the proposed method.

The received signal can be written in a vector form as
\begin{align}
\boldsymbol{y}\triangleq \operatorname{vec}\{\boldsymbol{Y}\}=(\boldsymbol{I}\otimes \boldsymbol{GA})\boldsymbol{s}+\boldsymbol{n},
\end{align} 
where $\boldsymbol{G}\triangleq \operatorname{diag}\{\boldsymbol{g}\}$. Therefore, with the DOA parameter $\boldsymbol{\theta}\triangleq\begin{bmatrix}
\theta_0,\theta_1,\dots,\theta_{K-1}
\end{bmatrix}^{\text{T}}$ and the gain-phase error $\boldsymbol{g}\triangleq \begin{bmatrix}
g_0,g_1,\dots,g_{N-1}
\end{bmatrix}^{\text{T}}$, and the received signal follows the Gaussian distribution
\begin{align}
\boldsymbol{y}\sim \mathcal{CN}(\boldsymbol{0},\boldsymbol{C}),
\end{align}
where $\boldsymbol{C}\triangleq (\boldsymbol{I}\otimes \boldsymbol{GA})\boldsymbol{B}[\boldsymbol{I}\otimes (\boldsymbol{GA})^{\text{H}}]+\sigma^2_{\text{n}}\boldsymbol{I}$.
The probability density function of Gaussian distribution $\boldsymbol{y}\sim \mathcal{CN}(\boldsymbol{0},\boldsymbol{C})$ can be expressed as
\begin{align}
f(\boldsymbol{x})=\frac{1}{\pi^N \det\{\boldsymbol{C}\}}e^{-\boldsymbol{y}^{\text{H}}\boldsymbol{C}^{-1}\boldsymbol{y}}.
\end{align}

The Fisher information matrix $\boldsymbol{F}$ can be written as
\begin{align}
\boldsymbol{F}\triangleq 
\begin{bmatrix}
\boldsymbol{F}_{1,1}&\boldsymbol{F}_{1,2}\\
\boldsymbol{F}_{2,1}&\boldsymbol{F}_{2,2}
\end{bmatrix},
\end{align}
where we have
\begin{align}
\boldsymbol{F}_{1,1}&=-\mathcal{E}\left\{\left.
\frac{\partial \ln f(\boldsymbol{y};\boldsymbol{\theta},\boldsymbol{g})}{\partial \boldsymbol{\theta} \partial \boldsymbol{\theta}}\right|\boldsymbol{\theta},\boldsymbol{g}
\right\},\\
\boldsymbol{F}_{1,2}&=-\mathcal{E}\left\{\left.
\frac{\partial \ln f(\boldsymbol{y};\boldsymbol{\theta},\boldsymbol{g})}{\partial \boldsymbol{\theta} \partial \boldsymbol{g}}\right|\boldsymbol{\theta},\boldsymbol{g}
\right\},\\
\boldsymbol{F}_{2,1}&=-\mathcal{E}\left\{\left.
\frac{\partial \ln f(\boldsymbol{y};\boldsymbol{\theta},\boldsymbol{g})}{\partial \boldsymbol{g} \partial \boldsymbol{\theta}}\right|\boldsymbol{\theta},\boldsymbol{g}
\right\},\\
\boldsymbol{F}_{2,2}&=-\mathcal{E}\left\{\left.
\frac{\partial \ln f(\boldsymbol{y};\boldsymbol{\theta},\boldsymbol{g})}{\partial \boldsymbol{g} \partial \boldsymbol{g}}\right|\boldsymbol{\theta},\boldsymbol{g}
\right\}.
\end{align}
The entries of Fisher information matrix $\boldsymbol{F}$ are given in Appendix~\ref{entry}.

The CRLB of DOA estimation can be expressed as 
\begin{align}
\operatorname{var}\{\boldsymbol{\theta}\}\geq\sum_{k=0}^{K-1} \left[\boldsymbol{F}^{-1}\right]_{k,k}.
\end{align}
However, in the parameter estimation problems, when the dimension of the unknown parameter is high, the FIM will be singular or very nearly so, especially in the case with sparse reconstruction, where the number of samples is much less than that in the oversampling scenario. The derivation of CRLB using $\boldsymbol{F}^{-1}$ can  be only obtained by assuming that the FIM is positive defined~\cite{4838872}. In our problem of sparse estimation with unknown gain-phase errors, the Fisher information matrix is singular, and it is inconvenient to obtain the inverse of the Fisher information matrix,  so we use a lower bound of FIM to describe the estimation performance as~\cite{2241683}
 \begin{align}
\operatorname{var}\{\boldsymbol{\theta}\}\geq\sum_{k=0}^{K-1} F^{-1}_{k,k}.
\end{align}

\section{Simulation Results}\label{sec6}
\begin{table}[!t]
	% increase table row spacing, adjust to taste
	\renewcommand{\arraystretch}{1.3}
	\caption{Simulation Parameters}
	\label{table1}
	\centering
	\begin{tabular}{cc}
		\hline
		\textbf{Parameter} & \textbf{Val1ue}\\
		\hline
		The signal-to-noise ratio (SNR) of received signal & $ 20 $ dB\\
		The number of pulses $P$ & $5$\\ 
		The number of antennas $N$ & $10$\\
		The number of signals $K$& $3$\\
		The space between antennas $d$& $0.5$ wavelength\\ 
		The detection DOA range & $\left[\ang{-70},\ang{70}\right]$\\
		The standard deviation of gain error $\sigma_{\text{A}}$ & $0.15$\\
		The standard deviation of phase error $\sigma_{\text{P}}$ & $10$ in degree\\ 
		\hline
	\end{tabular}
\end{table}

The simulation parameters are given in Table~\ref{table1}. The number of Monte Carlo simulations  is $10^3$. We consider the DOA estimation in the scenario with much few measurements (snapshots) $P=5$. The minimum separation between signals in degree is $\Delta \geq \ang{10}$. All the simulation results are obtained on a PC with Matlab R2018b with a 2.9 GHz Intel Core i5 and 8 GB of RAM. The code of proposed algorithm will be available online after that the paper is accepted.

The gain errors among antennas are generated by a Gaussian distribution 
\begin{align}
g_{n} \sim\mathcal{N}(0,\sigma^2_{\text{A}}),\quad n=0,1,\dots, N-1,
\end{align}
where $\sigma^2_{\text{A}}$ denotes the variance of gain errors.
The phase errors in degree also follow a Gaussian distribution 
\begin{align}
\phi_{n} \sim\mathcal{N}(0,\sigma^2_{\text{P}}),\quad n=0,1,\dots, N-1,
\end{align}
where $\sigma^2_{\text{P}}$ denotes the variance of phase errors.
Then, the normalized gain for the $n$-th antenna with gain-phase error is $(1+g_{n})e^{j \phi_{n}}$. 
Hence, the parameter $C_{\text{e}}$ can be chosen as the one with $C_e^2 \geq N(\sigma^2_{\text{A}}+\sigma^2_{\text{P}})=0.514$.

\begin{table} 
	% increase table row spacing, adjust to taste
	\renewcommand{\arraystretch}{1.3}
	\caption{DOA Estimation}
	\label{table2}
	\centering
	\begin{tabular}{ccccc}
		\hline
		\textbf{Methods} & \textbf{Signal $1$} &  \textbf{Signal $2$} & \textbf{Signal $3$}&RMSE (deg)\\
		\hline
		Ground-truth &  $\ang{-56.8889}$ & $\ang{-7.6806}$ &  $\ang{5.9595}$ &-- \\
		ANM & $\ang{-56.4480}$ &$\ang{5.6840}$&$\ang{28.7000}$&$232$\\ 
		MUSIC & $\ang{-56.4900}$&$\ang{-7.6020}$&$\ang{6.0900}$&$0.06076$\\
		SOMP  & $\ang{ -56.3640}$&$\ang{-8.2460}$&$\ang{5.7400}$&$0.2144$\\
		SBL & $\ang{-56.0000}$&$\ang{-7.0000}$&$\ang{5.6000}$&$0.4608$\\ 
		Proposed method & $\ang{-56.6860}$&$\ang{-7.6300}$&$\ang{5.9640}$ & $0.01458$\\
		\hline
	\end{tabular}
\end{table}

First, we try to estimate $3$ signals from the received signals, and the ground-truth DOAs are  $\ang{-56.8889}$, $\ang{-7.6806}$, and $\ang{5.9595}$. When the ANM method is adopted, the DOAs are estimated by the polynomial of the ANM method. As shown in Fig.~\ref{ANM}, the polynomial of ANM method is given. Since the antennas in the array have gain-phase errors, the polynomial has multiple peak values, and the DOAs cannot be estimated well. The estimated DOAs are  $\ang{-56.4480}$, $\ang{5.6840}$, and $\ang{28.7000}$, so the estimation error is much large. However, when the proposed method with GP-ANM is adopted, we can obtain the polynomial in Fig.~\ref{poly}. The peak values are well distinguished, and the estimated DOAs are $\ang{-56.6860}$, $\ang{-7.6300}$, and $\ang{5.9640}$. Therefore, the proposed method outperforms the traditional ANM method in the DOA estimation with gain-phase errors.

\begin{figure}
	\centering
	\includegraphics[width=3.2in]{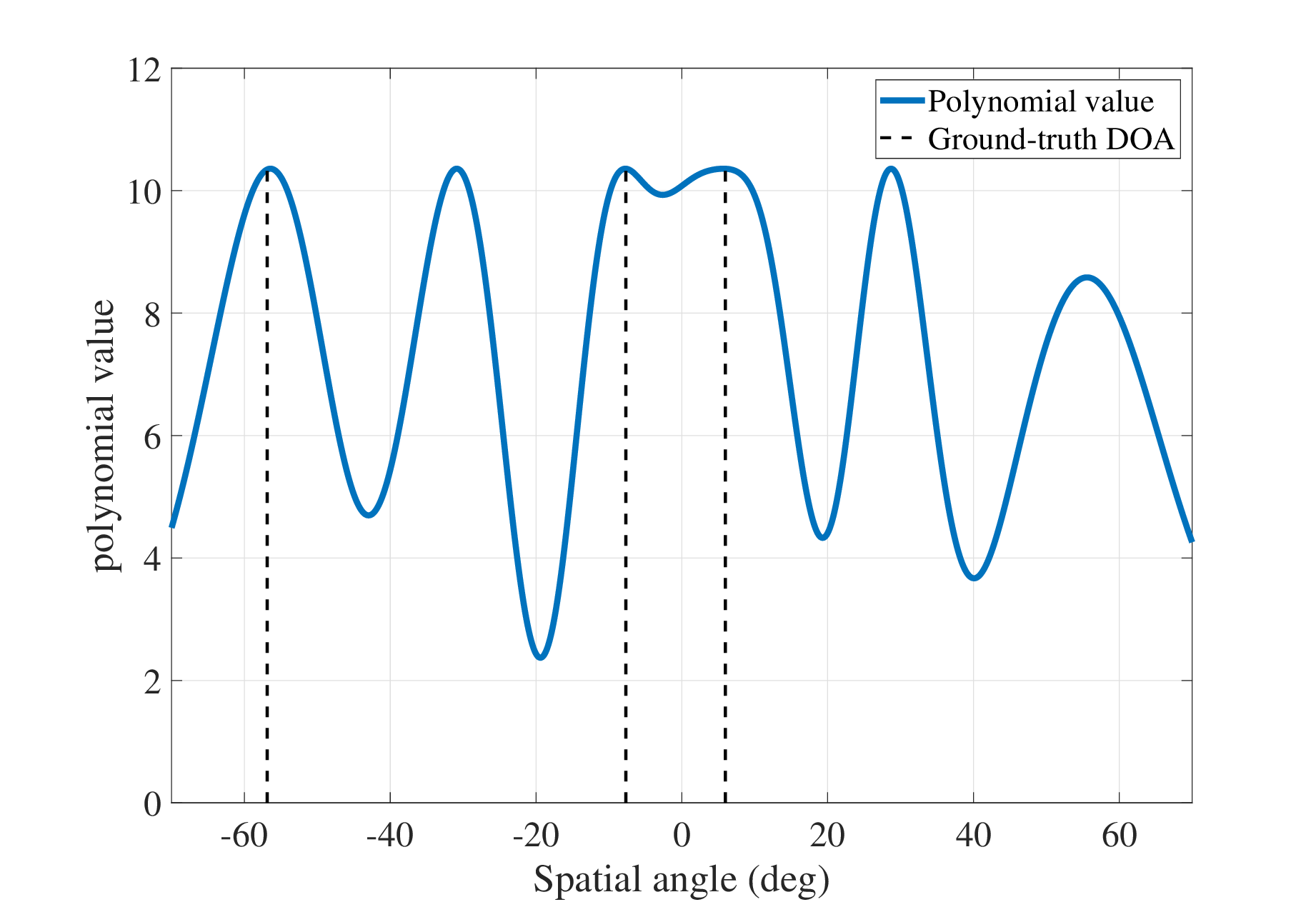}
	\caption{The polynomial in ANM method.}
	\label{ANM}
\end{figure} 

\begin{figure}
	\centering
	\includegraphics[width=3in]{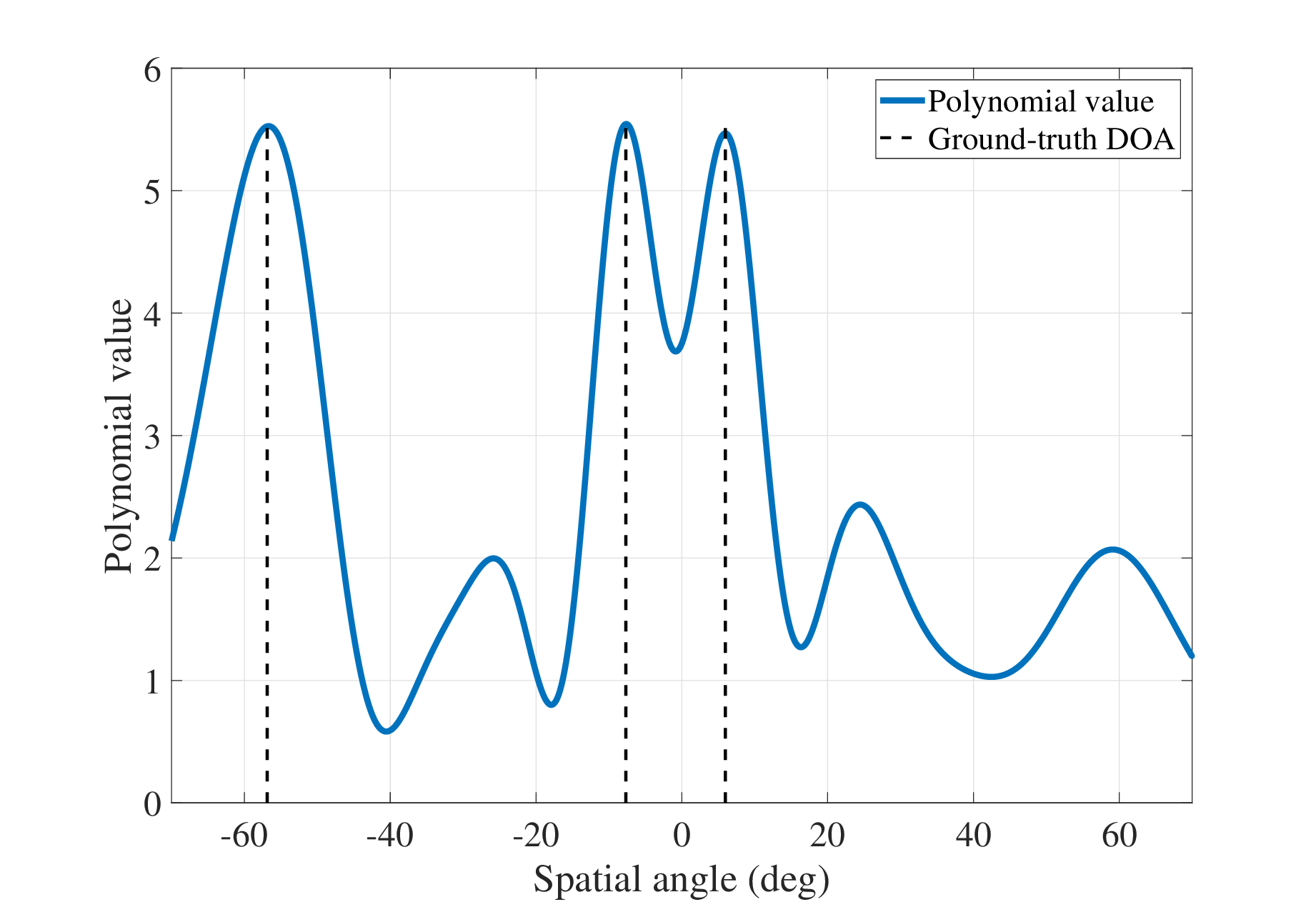}
	\caption{The polynomial in the proposed method.}
	\label{poly}
\end{figure}

\begin{figure}
	\centering
	\includegraphics[width=2.5in]{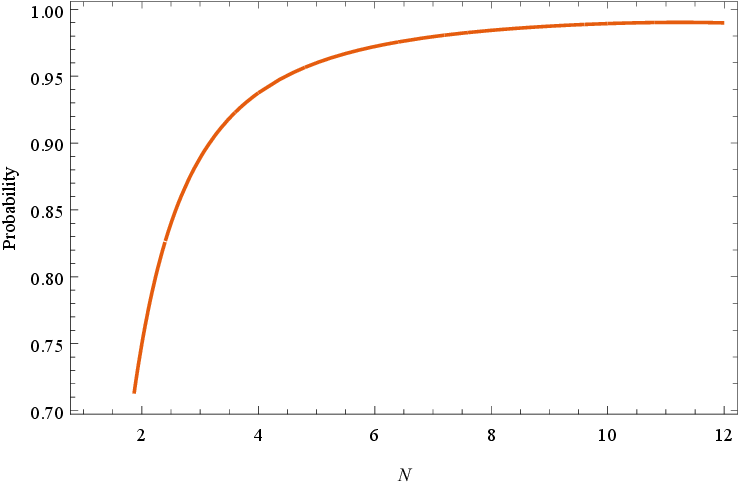}
	\caption{The probability for signal reconstruction in (\ref{eq35}).}
	\label{NN}
\end{figure} 

 With the simulation parameters in Table~\ref{table1}, the probability of signal reconstruction is shown in Fig.~\ref{NN} with the different number of antennas, and this figure is different from the direct performance of DOA estimation, such as the root mean square error (RMSE) of DOA estimation. When the number of antennas increases, the probability that the sparse reconstruction signal can approach the ground-truth signal is also improved. Therefore, a high probability that the sparse signal can be reconstructed with limited error can be achieved by selecting an appropriate regularization parameter $\tau$.

Additional, we compare the DOA estimation performance of the proposed method with existing methods, including MUSIC, simultaneous orthogonal matching pursuit (SOMP)~\cite{7738592}, and sparse Bayesian learning (SBL)~\cite{8667328} methods. MUSIC method is a subspace-based method and has been widely used in the DOA estimation with better performance and robustness. SOMP method is the sparse-based method and has been widely used in the sparse reconstruction problem. SBL method is a sparse method and has great reconstruction performance but has high computational complexity. The DOA estimation performance is measured by the RMSE. RMSE is defined as
\begin{align}
    \text{RMSE}\triangleq \sqrt{\frac{1}{KN_{\text{mc}}}\sum_{n_{\text{mc}}=0}^{N_{\text{mc}}-1}\sum_{k=0}^{K-1}\left(\theta_{n_{\text{mc},k}}-\hat{\theta}_{n_{\text{mc},k}}\right)^2},
\end{align}
where $N_{\text{mc}}$ denotes the number of Monte Carlo simulations, and $K$ denotes the number of signals in one simulation. $\theta_{n_{\text{mc},k}}$ is the ground-truth DOA of the $k$-th signal during the $n_{\text{mc}}$-th simulation, and $\hat{\theta}_{n_{\text{mc},k}}$ is the corresponding estimated DOA. In this paper, we assume that the number of signals can be estimated precisely using the traditional methods, such as Akaike information theoretic criteria (AIC) and minimum description length (MDL)~\cite{1164557,330365,995060}. The RMSEs of ANM, MUSIC, SOMP, SBL and the proposed method are shown in Table.~\ref{table1}. The RMSE of proposed method is $0.01458$ in deg, and $76\%$ better than MUSIC. Additionally, since the multiple peak values in the polynomial of ANM method, the DOA cannot be estimated well and the RMSE of  ANM method is much larger than other methods. In the SBL method, the spatial angle is discretized into grids with the grid size being $\ang{0.5}$ to have a comparable computational time with the proposed method.  The spatial spectrums of these $4$ methods are shown in Fig.~\ref{spectrum}, where we can see that the spectrum of SBL is much better than that of MUSIC method. SOMP and proposed methods are the sparse-based method, so we show the reconstruction results in the figure of spatial spectrum. The spatial spectrum of proposed method is much close to the ground-truth DOA.

\begin{figure}
	\centering
	\includegraphics[width=3in]{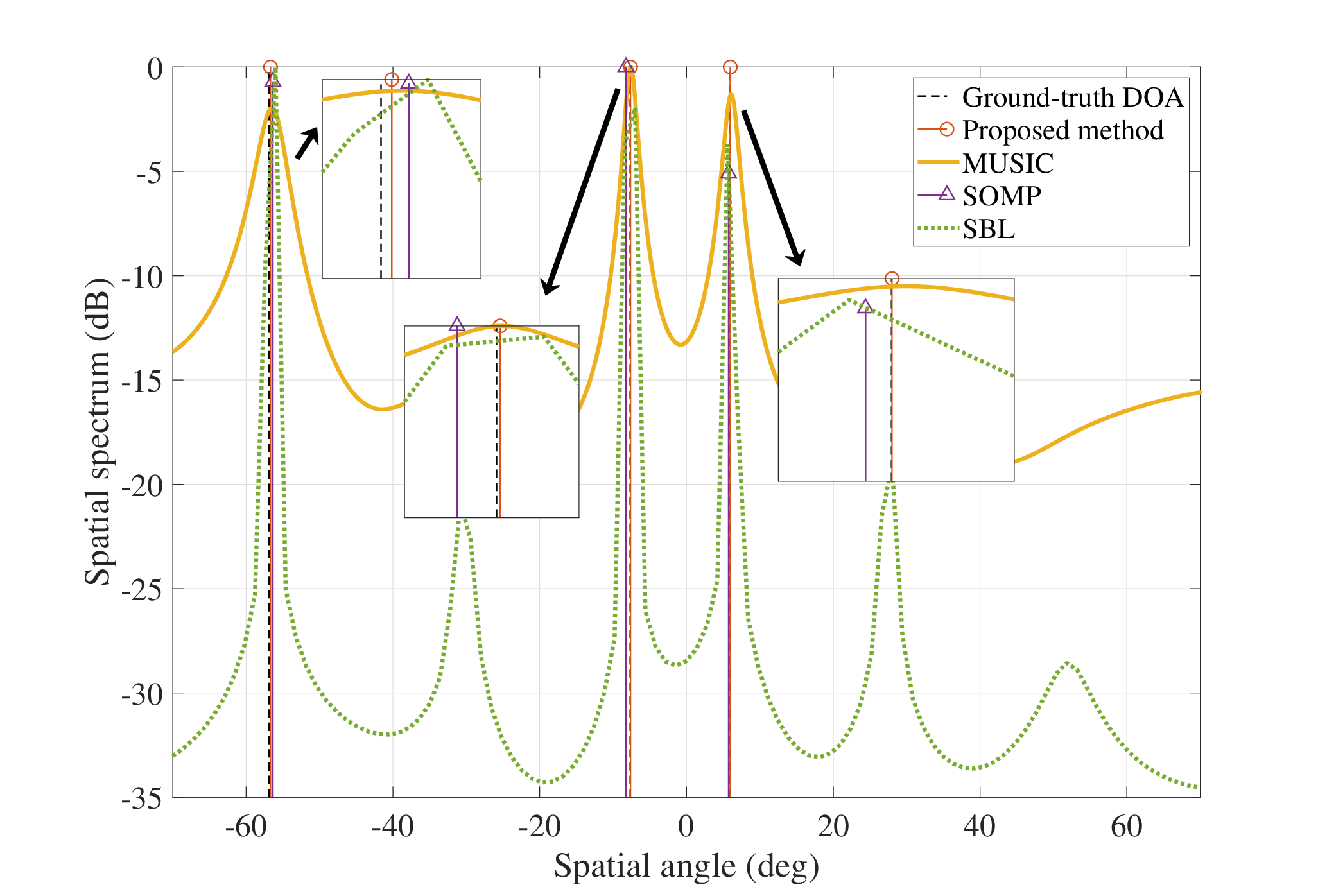}
	\caption{The spatial spectrum for DOA estimation.}
	\label{spectrum}
\end{figure} 

\begin{figure}
	\centering
	\includegraphics[width=3in]{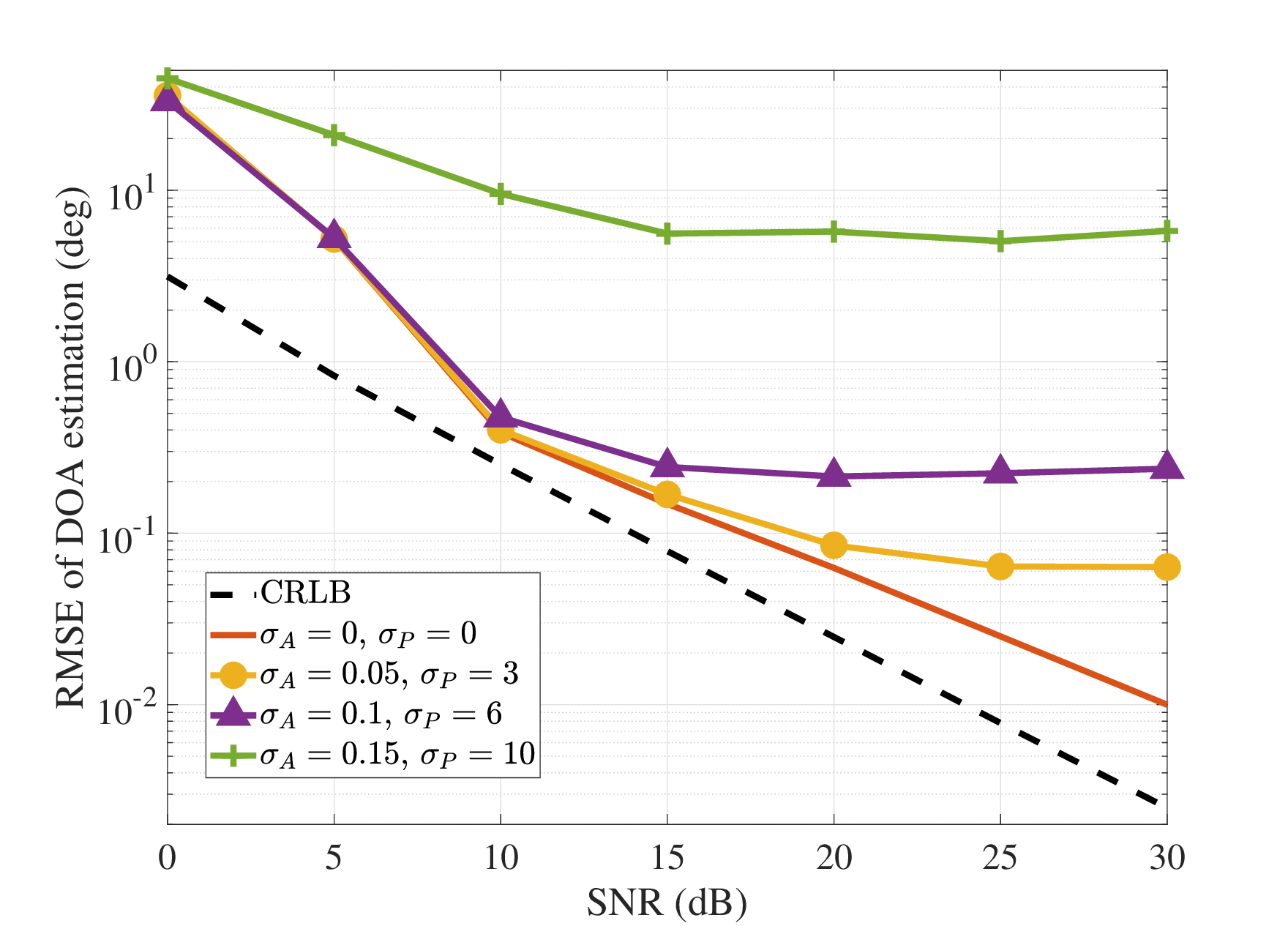}
	\caption{The DOA estimation with different gain-phase errors using the proposed method (ANM).}
	\label{crlbANM}
\end{figure}

\begin{figure}
	\centering
	\includegraphics[width=3in]{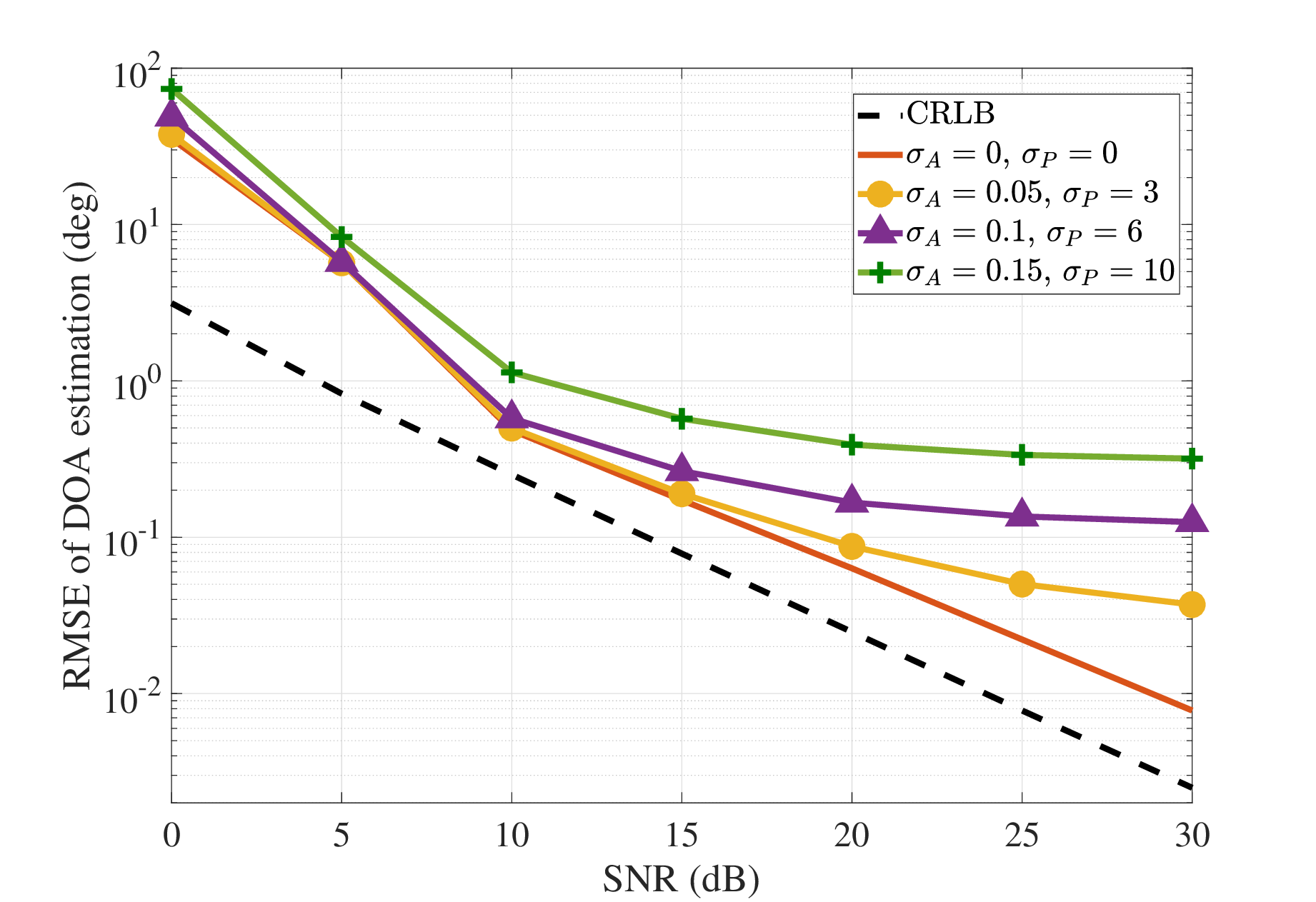}
	\caption{The DOA estimation with different gain-phase errors using the proposed method (Proposed method).}
	\label{crlb}
\end{figure}

Then, to show the DOA estimation performance with different variances of grain-phase errors, we give the DOA estimation performance with different variances in Fig.~\ref{crlbANM} and Fig.~\ref{crlb}, where Fig.~\ref{crlbANM} uses the traditional ANM method and Fig.~\ref{crlb} uses the proposed method.  When the variance of the gain-phase error is small, both  ANM and proposed methods can approach the CRLB in DOA estimation. However, when the variance of the gain-phase error is large, the ANM method degrades the RMSE significantly. The proposed method can also keep the estimation performance well. Therefore, with the GP-ANM, the effect of gain-phase error can be reduced effectively.

\begin{figure}
	\centering
	\includegraphics[width=3in]{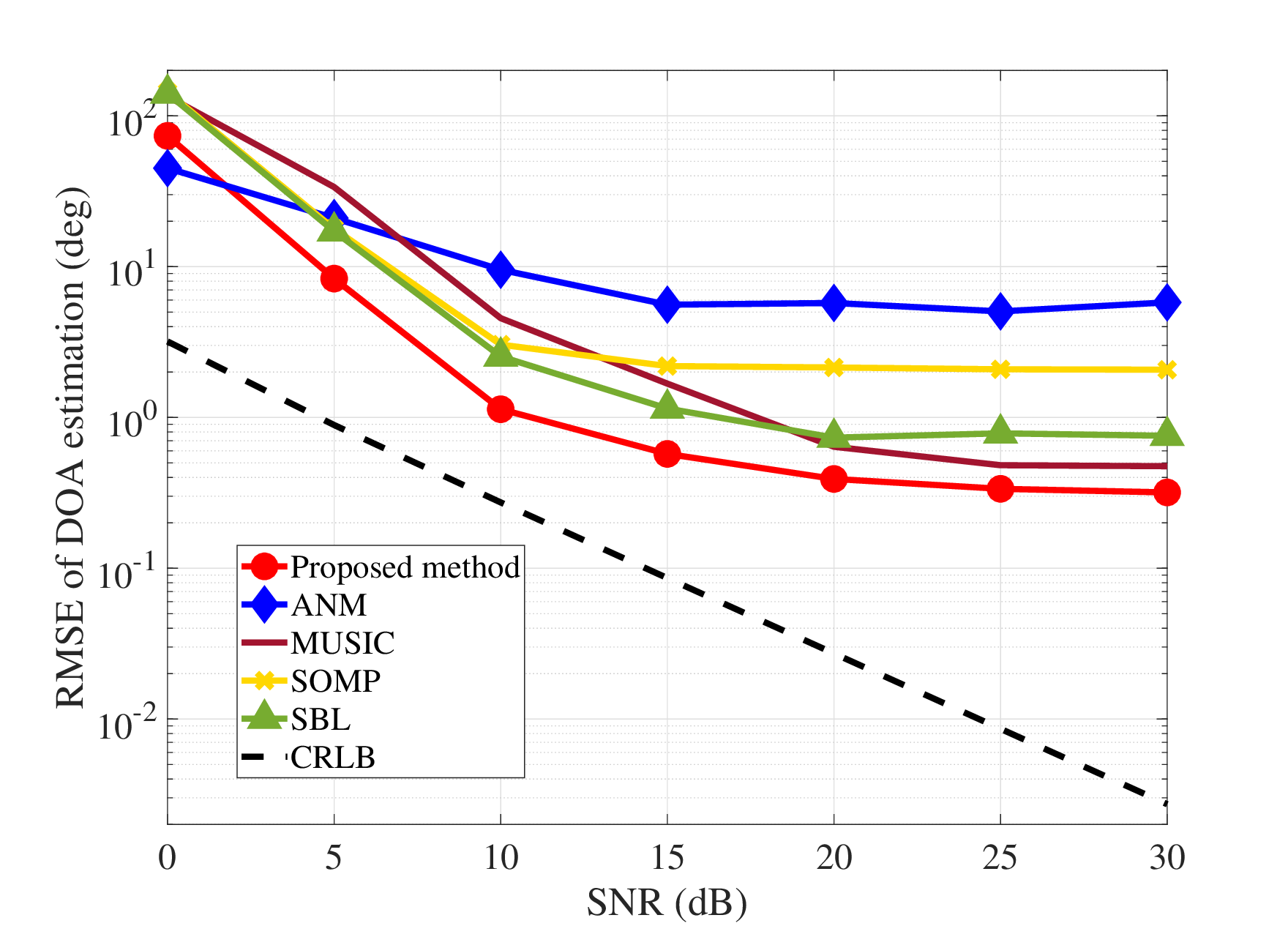}
	\caption{The DOA estimation with different SNRs.}
	\label{SNR}
\end{figure} 

\begin{figure}
	\centering
	\includegraphics[width=3.2in]{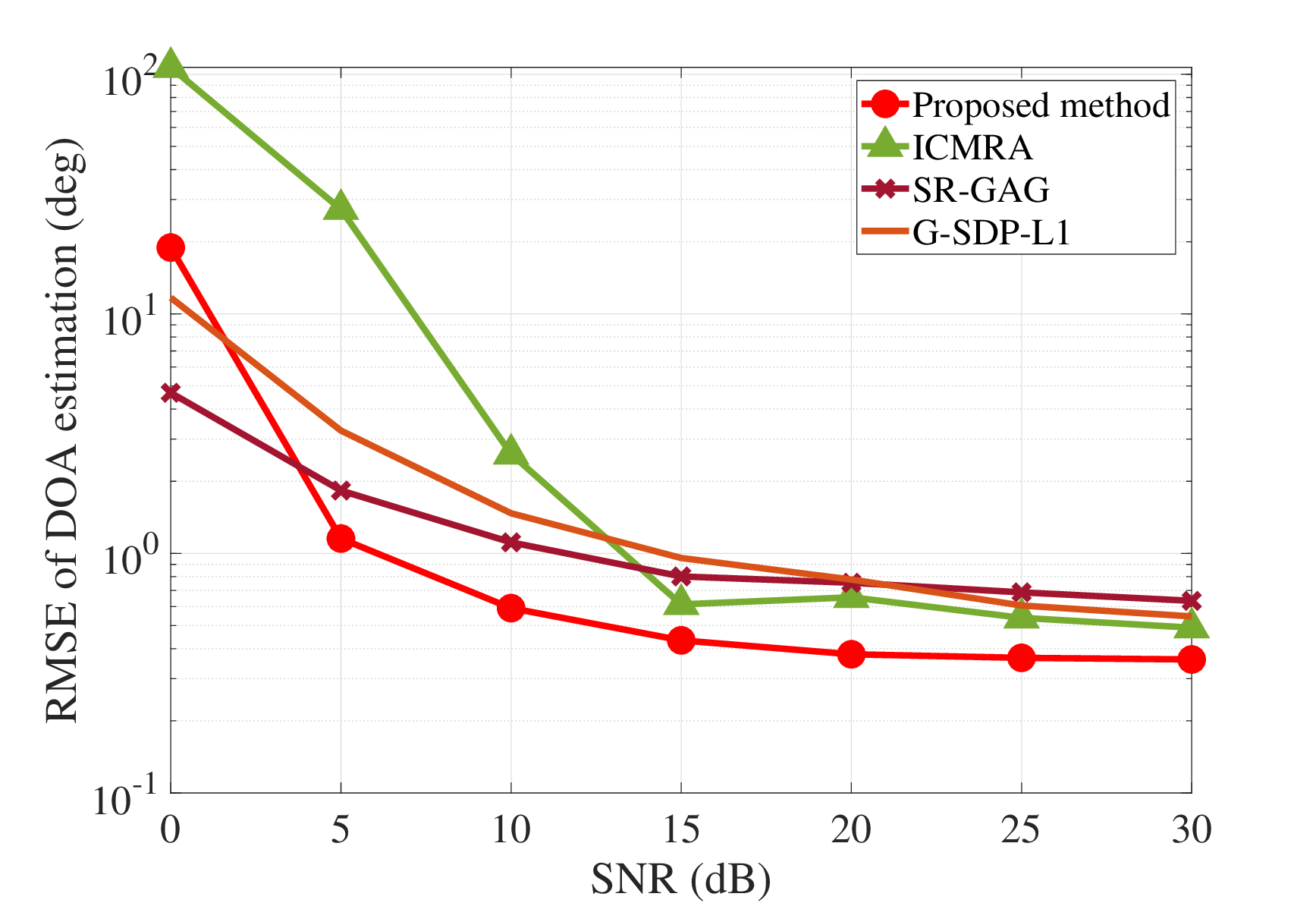}
	\caption{The DOA estimation performance compared with ICMRA, SR-GAG and G-SDP-L1 methods.}
	\label{icmra}
\end{figure}

For different SNRs, the DOA estimation performance is shown in Fig.~\ref{SNR}, where the SNR is from $0$ dB to $30$ dB. When the SNR is higher than $15$ dB, the estimation performance is almost the same. MUSIC method can achieve better estimation performance than the ANM, SOMP, and SBL methods in the scenario with gain-phase errors. The proposed method achieves the best estimation performance among these methods when the SNR is higher than $5$ dB. Since the CRLB does not consider the gain-phase error, the CRLB can be further improved with higher SNR, but the estimation performance has platform effect and cannot be improved when SNR is higher than $20$ dB. Moreover, as shown in Fig.~\ref{SNR}, when the SNR of the received signal is $20$ dB, the RMSEs of the DOA estimation using the ANM method, the SOMP method, the SBL method, the MUSIC method and the proposed method are $\ang{5.741}$, $\ang{2.144}$, $\ang{0.735}$ $\ang{0.639}$ $\ang{0.391}$, respectively. With the gain-phase errors, the ANM method cannot estimate the DOA accurately, but the MUSIC method as a robust method can achieve a higher DOA resolution than the ANM method.  Compared with the ANM method, the proposed method can improve the DOA resolution about $\ang{5.35}$ in the scenario with the standard derivation of gain error being $\sigma_{\text{A}}=0.15$ and the that of phase error being $\sigma_{\text{P}}=10$ in degree. 
Additionally, the DOA estimation performance of the proposed method is also compared with the improved covariance matrix reconstruction approach (ICMRA)~\cite{8320855}, soft recovery approach for general antenna geometries (SR-GAG)~\cite{barzegar2017estimation}, and a generalization of SDP formulation of $\ell_1$ norm optimization problem (G-SDP-L1)~\cite{16M1071730}. The ICMRA method is based on the low-rank reconstruction, where a covariance matrix of the received signals is used for the DOA estimation. In the simulation section, the number of antennas is $10$ and the snapshots are $5$, so the covariance matrix cannot be accurately estimated. The SR-GAG method is proposed for a general antenna geometry. When this method is applied to the system model considered in this paper, the method will be the same with the ANM method, since the antenna geometry is ULA. The G-SDP-L1 method is a general case of gauge function and atomic norm, but this extension cannot describe the gain-phase errors well. Moreover, these methods have not considered the gain-phase errors in the system model. Therefore, better performance can be achieved by the proposed method.

\begin{figure}
	\centering
	\includegraphics[width=3.2in]{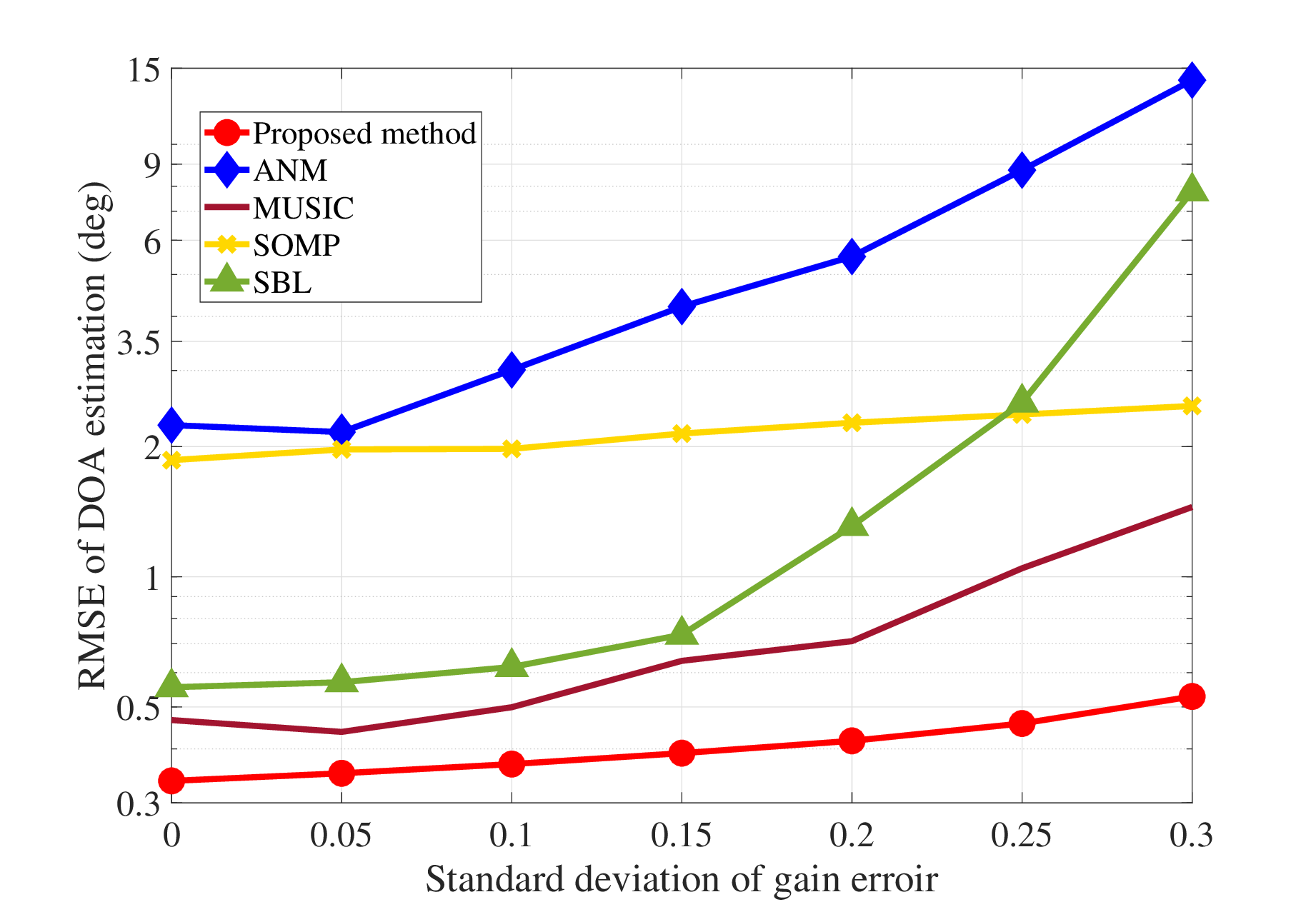}
	\caption{The DOA estimation with different gain errors.}
	\label{amp}
\end{figure} 

\begin{figure}
	\centering
	\includegraphics[width=3in]{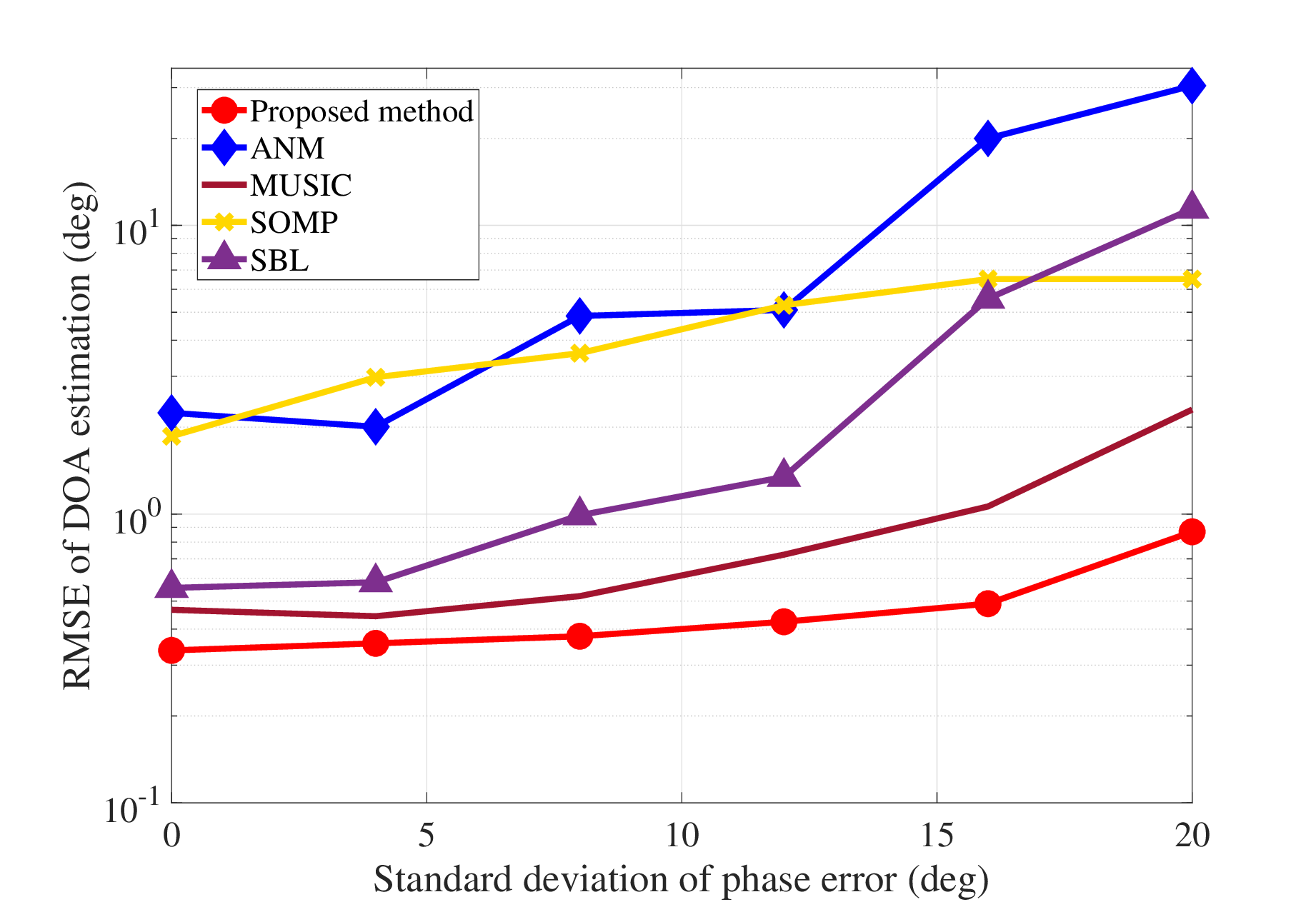}
	\caption{The DOA estimation with different phase errors.}
	\label{phase}
\end{figure}

\begin{figure}
		\centering
		\includegraphics[width=3in]{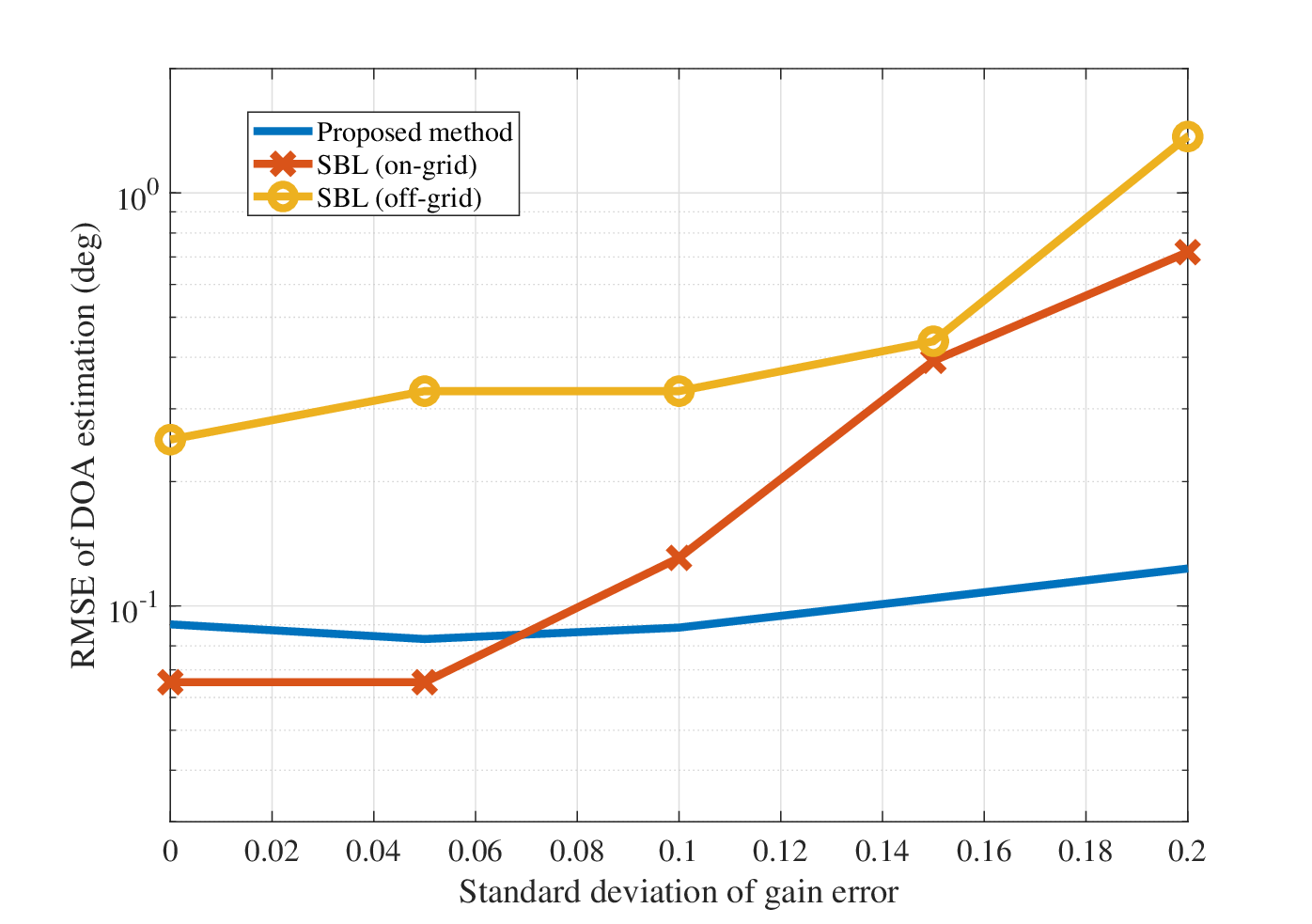}
		\caption{The DOA estimation performance compared with the SBL method.}
		\label{sbl}
	\end{figure} 

When we keep the standard deviation of phase errors $\sigma_{\text{P}}$ as $10$ in degree, and change that of gain errors, the corresponding RMSE of DOA estimation is shown in Fig.~\ref{amp}. $\sigma_{\text{A}}$ changes from $0$ to $0.3$, and the estimation error is only improved from $0.33$ to $0.32$ using the proposed method. However, the existing methods degrade the DOA estimation performance significantly with larger gain errors. Moreover, keeping $\sigma_{\text{A}}=0.15$, we change $\sigma_{\text{P}}$ from $0$ to $20$ in degree, and the DOA estimation performance is shown in Fig.~\ref{phase}. As shown in this figure, the proposed method achieves the best estimation performance among these methods. Therefore, in the scenario with gain-phase errors, the proposed method can work well. Moreover, with only the gain errors, the DOA estimation performance of the proposed method is also compared with that of the SBL method, as shown in Fig.~\ref{sbl}. 
``SBL (on-grid)'' is the DOA estimation performance with the signal angles being at the discretized angles exactly in the spatial domain, and ``SBL (off-grid)'' means that the signals can be not precisely at the discretized angles. As shown in Fig.~\ref{sbl}, when the signals are on-grid, the SBL method outperforms the proposed method in the scenario with small gain-phase errors.  However, in the scenario with large gain-phase errors or the off-grid signals, the estimation performance of the SBL method is worse than that of the proposed method.

\begin{figure}
	\centering
	\includegraphics[width=3in]{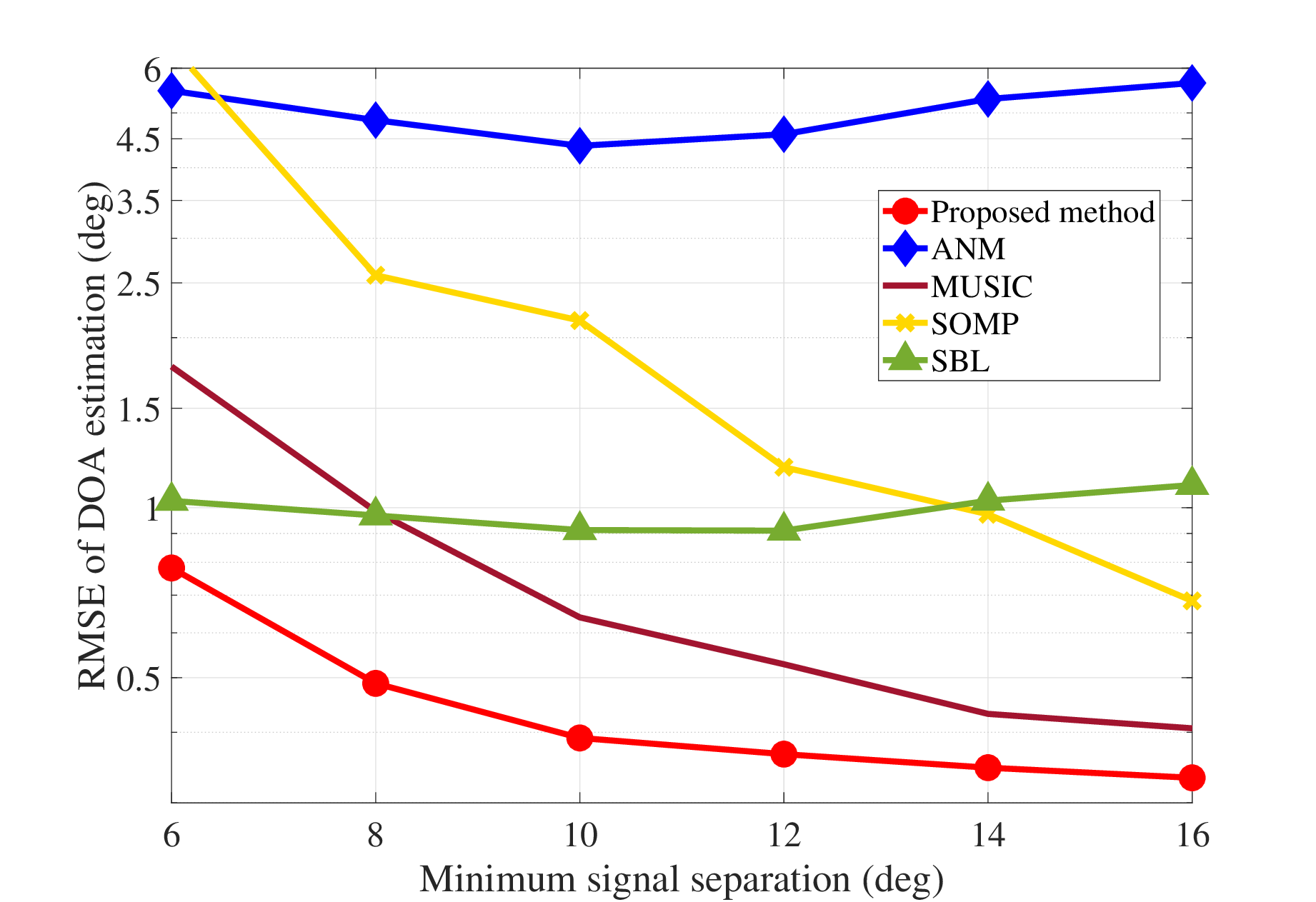}
	\caption{The DOA estimation with different minimum signal separations.}
	\label{separation}
\end{figure} 

In the super-resolution methods, the minimum separation between signals is important and shows the ability of super-resolution, so we show the DOA estimation performance with different minimum separations in Fig.~\ref{separation}. With larger separation, the correlation between the received signal can be reduced so that the better estimation performance can be achieved. Additionally, the number of measurements is vital for the complexity consideration, and the corresponding estimation performance is shown in Fig.~\ref{P}. The proposed can outperform the existing methods when the number of measurements is more than $3$.

\begin{figure}
	\centering
	\includegraphics[width=3in]{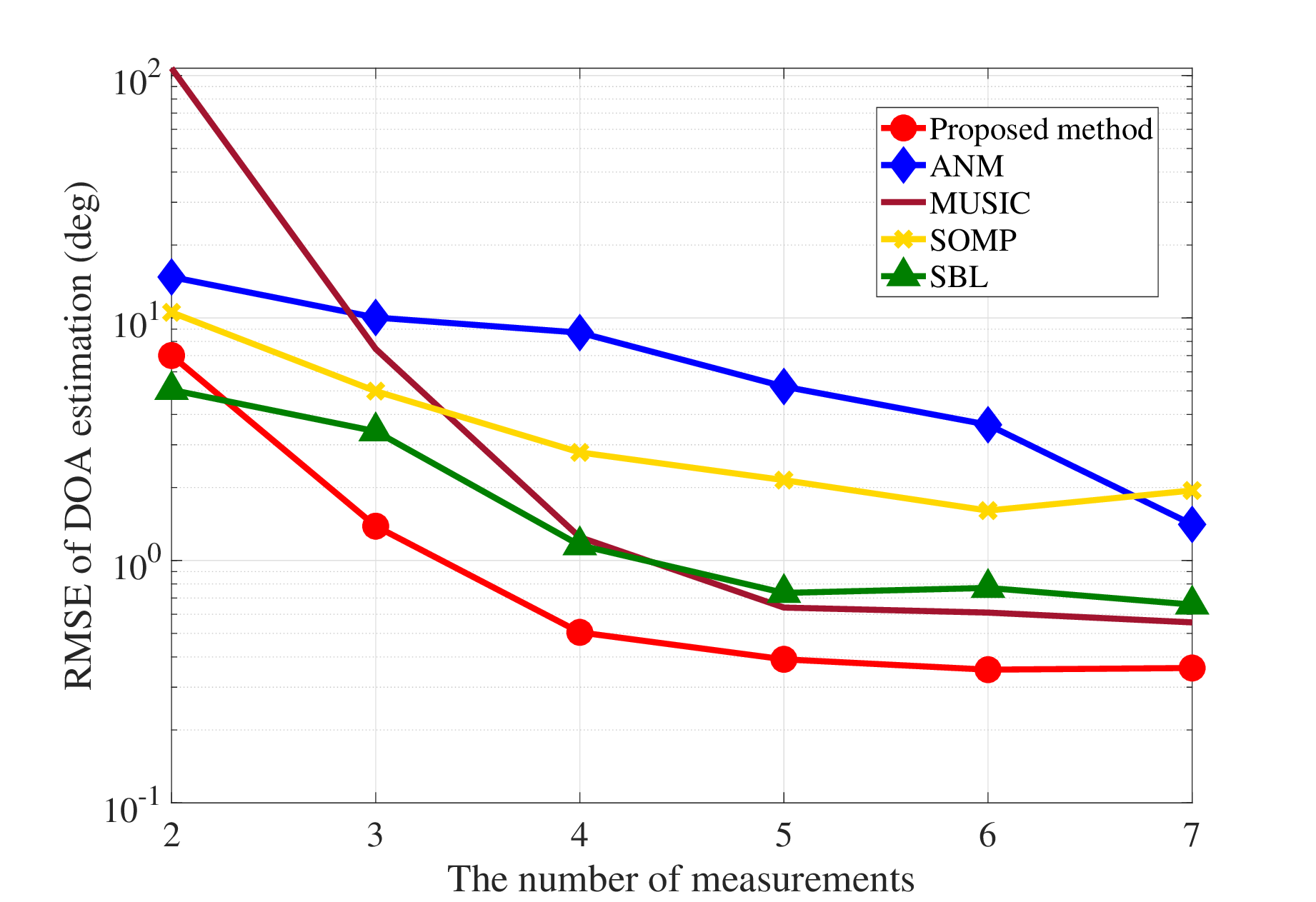}
	\caption{The DOA estimation with different numbers of measurements.}
	\label{P}
\end{figure} 

\begin{figure}
	\centering
	\includegraphics[width=3in]{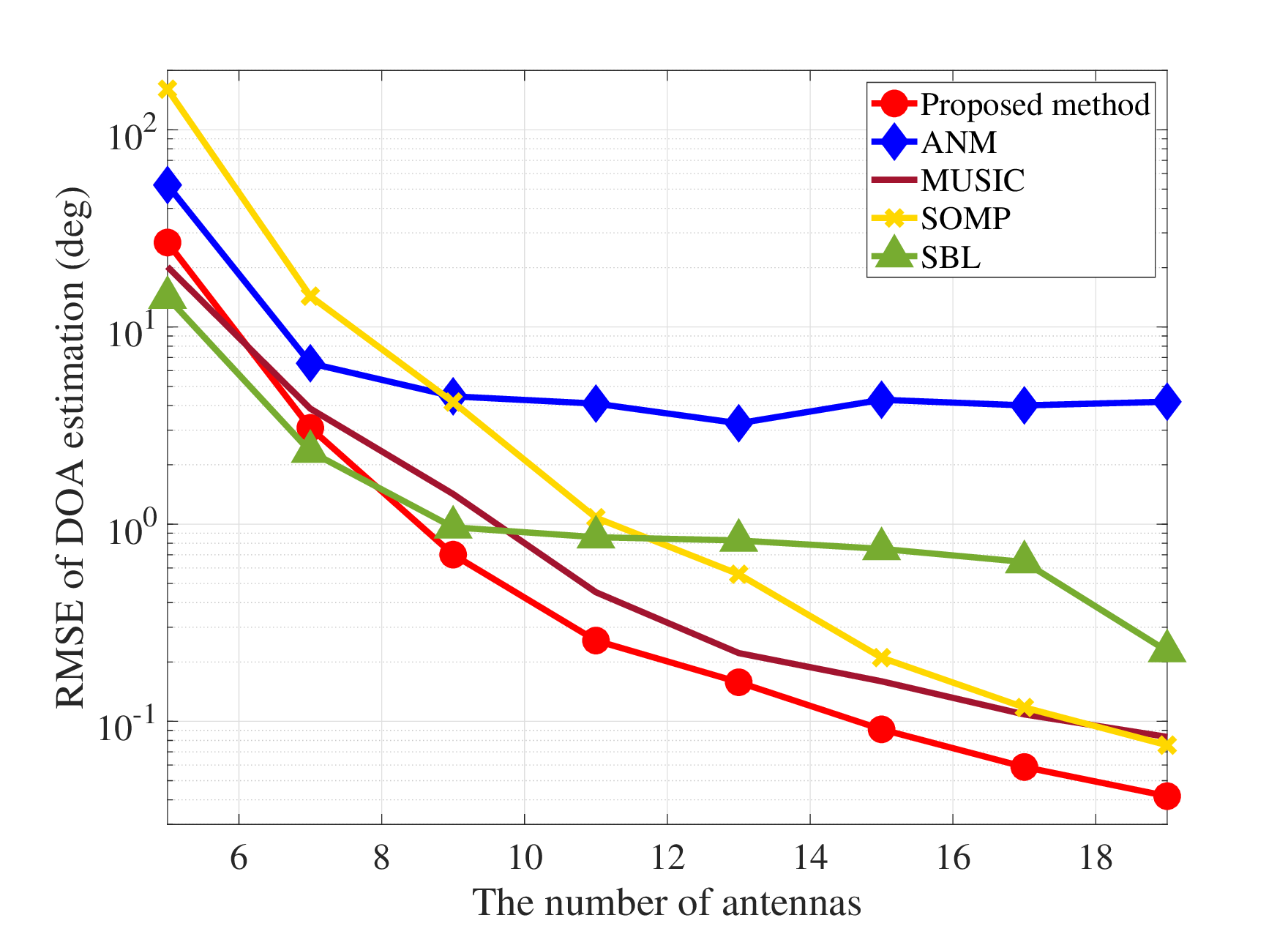}
	\caption{The DOA estimation with different numbers of antennas.}
	\label{N}
\end{figure} 

\begin{figure}
	\centering
	\includegraphics[width=3in]{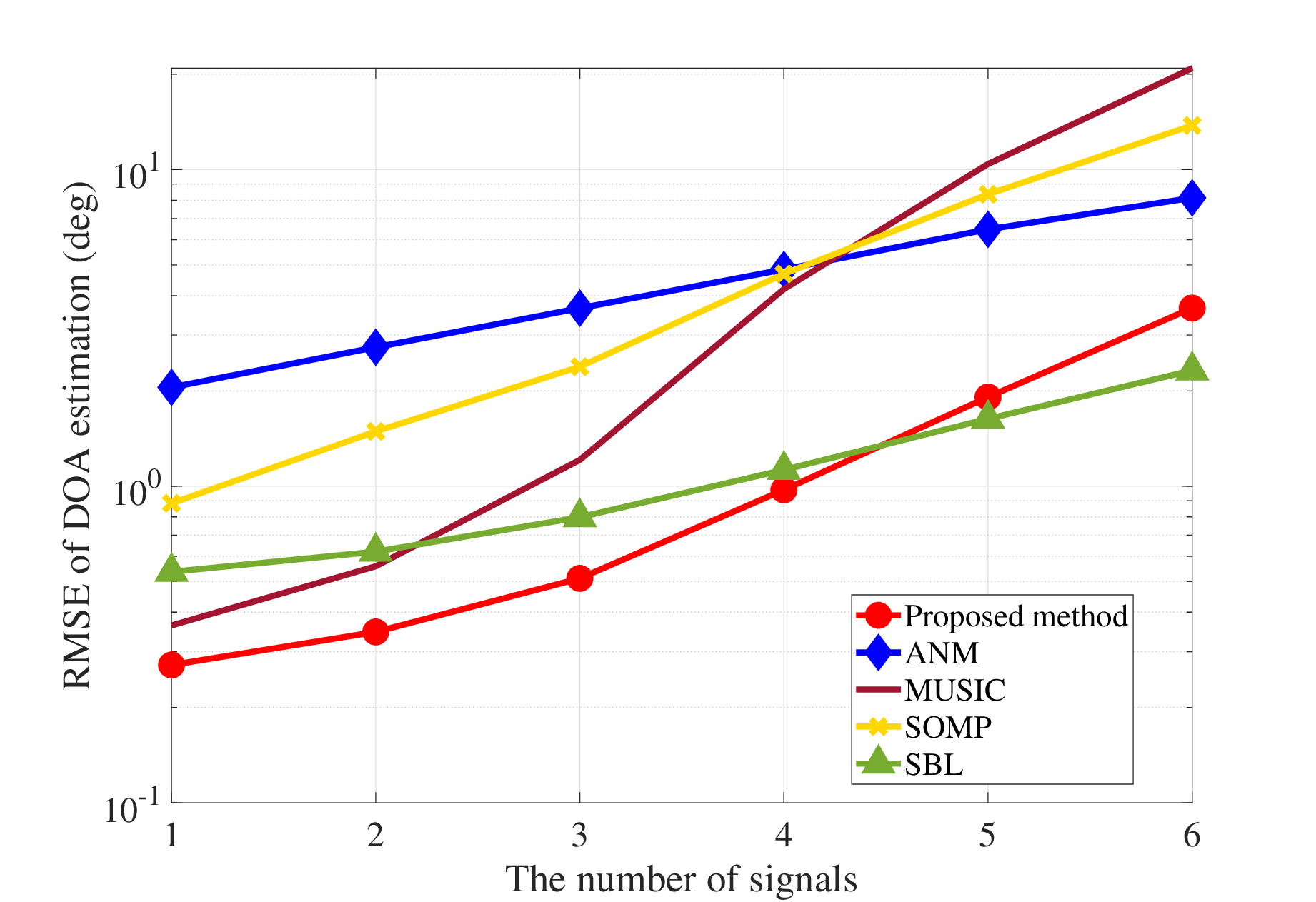}
	\caption{The DOA estimation with different numbers of signals.}
	\label{K}
\end{figure} 

With different numbers of antennas and signals, the DOA estimation performance is shown in Fig.~\ref{N} and Fig.~\ref{K}, respectively. As shown in these two figures, when the number of antennas is more than $10$ or the number of signals is less than $4$, the proposed method can achieve better estimation performance. Furthermore, the computational time is shown in Table~\ref{table3}, the proposed method has relative higher computational complexity. As shown in the results of this section, the proposed method can generally achieve better DOA estimation performance in the scenario with gain-phase errors.

\begin{table}[!t]
	% increase table row spacing, adjust to taste
	\renewcommand{\arraystretch}{1.3}
	\caption{Computational Time}
	\label{table3}
	\centering
	\begin{tabular}{cccccc}
		\hline
		  & \textbf{ANM} &  \textbf{MUSIC} & \textbf{SOMP}&\textbf{SBL}& \textbf{Proposed method}\\
		\hline 
		\textbf{Time (s)} & $2.3047$ &$ 0.0319 $ &$ 0.0501 $ &$0.0780$ & $1.0840$ \\ 
		\hline
	\end{tabular}
\end{table}

\section{Conclusions}\label{sec7}
The  DOA estimation problem has been considered in the scenario with gain-phased errors, and the GP-ANM has been proposed to formulate the optimization problem. Then, the SDP formulation has been derived to solve the DOA estimation problem efficiently, and the corresponding regularization parameter has been obtained theoretically. Simulation results show that the proposed DOA estimation method outperforms the existing methods in the scenario with gain-phase errors. Future work will focus on the generalized atomic norm in the applications with imperfect antennas.

\appendices
\section{The SDP Proof for the Optimization Problem (\ref{sdp})}\label{sdpAP}
First, in the constraint $\operatorname{Tr}(\boldsymbol{Q}) + (C_{\text{e}}+2\sqrt{N})C_{\text{e}} \left\|\boldsymbol{Q}\right\|_2  -1 \leq 
	0$, both the norm operation $\left\|\boldsymbol{Q}\right\|_2 $  and the trace operation $\operatorname{Tr}(\boldsymbol{Q})$ are convex functions, so this constraint is a convex constraint. The optimization problem (\ref{sdp}) is a convex optimization problem.
	
Then, to show that the constraint $\operatorname{Tr}(\boldsymbol{Q}) + (C_{\text{e}}+2\sqrt{N})C_{\text{e}} \left\|\boldsymbol{Q}\right\|_2  -1 \leq 
	0$ is a SDP constraint, we can formulate a semidefinite matrix
	\begin{align} \label{eq72}		\begin{bmatrix}
			(1-\operatorname{Tr}(\boldsymbol{Q}))\boldsymbol{I} & (C_{\text{e}}+2\sqrt{N})C_{\text{e}} \boldsymbol{Q}^{\text{H}} \\ (C_{\text{e}}+2\sqrt{N})C_{\text{e}} \boldsymbol{Q} & (1-\operatorname{Tr}(\boldsymbol{Q}))\boldsymbol{I} 
		\end{bmatrix}\succeq 0,
	\end{align} 
	with the Schur complement theory, if and only if we have
	\begin{align}
    	& (1-\operatorname{Tr}(\boldsymbol{Q}))\boldsymbol{I} \succeq 0\\
    	& (1-\operatorname{Tr}(\boldsymbol{Q}))^2C_{\text{e}}^2\boldsymbol{I}- (C_{\text{e}}+2\sqrt{N})^2\boldsymbol{Q}^{\text{H}}  \boldsymbol{Q}\succeq 0.\label{eq74}
    \end{align}
From (\ref{eq74}), for arbitrary vector $\boldsymbol{t}$, we can obtain a function $f(\boldsymbol{t})\triangleq \boldsymbol{t}^{\text{H}}\left[(1-\operatorname{Tr}(\boldsymbol{Q}))^2\boldsymbol{I}- (C_{\text{e}}+2\sqrt{N})^2C_{\text{e}}^2\boldsymbol{Q}^{\text{H}}  \boldsymbol{Q}\right]\boldsymbol{t}\geq 0$, and formulate an optimization problem 
\begin{align}\label{eq75}
	\min_{\boldsymbol{t}}\quad &f(\boldsymbol{t})\\
	\text{s.t.}\quad &\|\boldsymbol{t}\|_2= C_{\text{t}},\notag
\end{align}
where $C_{\text{t}}$ is a positive constant. From (\ref{eq75}), the minimum value of $f(\boldsymbol{t})$ can be achieved as $f(\|\boldsymbol{t}\|_2\boldsymbol{q}_{\text{max}})$, where $\boldsymbol{q}_{\text{max}}$ is an eigenvector  corresponding to the maximum eigenvalue $\lambda_{\text{max}}$ of $\boldsymbol{Q}^{\text{H}}\boldsymbol{Q}$. Hence, (\ref{eq74}) is satisfied, if and only if,  for arbitrary $\boldsymbol{t}$, we have $f(\|\boldsymbol{t}\|_2\boldsymbol{q}_{\text{max}})\geq 0$.

Since  $\|\boldsymbol{Q}\|^2_2=\lambda_{\text{max}}$, we can simplify $f(\|\boldsymbol{t}\|_2\boldsymbol{q}_{\text{max}})$ as
\begin{align}\label{eq76}
	&f(\|\boldsymbol{t}\|_2\boldsymbol{q}_{\text{max}}) = \|\boldsymbol{t}\|_2\boldsymbol{q}_{\text{max}}^{\text{H}}(1-\operatorname{Tr}(\boldsymbol{Q}))^2\boldsymbol{I}\|\boldsymbol{t}\|_2\boldsymbol{q}_{\text{max}}\notag\\
	&\qquad-\|\boldsymbol{t}\|_2\boldsymbol{q}_{\text{max}}^{\text{H}} (C_{\text{e}}+2\sqrt{N})^2C_{\text{e}}^2\boldsymbol{Q}^{\text{H}}  \boldsymbol{Q}\|\boldsymbol{t}\|_2\boldsymbol{q}_{\text{max}}\\
&\quad = \|\boldsymbol{t}\|_2^2 (1-\operatorname{Tr}(\boldsymbol{Q}))^2-\|\boldsymbol{t}\|^2_2(C_{\text{e}}+2\sqrt{N})^2 C_{\text{e}}^2\boldsymbol{q}_{\text{max}}^{\text{H}}  \boldsymbol{Q}^{\text{H}}  \boldsymbol{Q} \boldsymbol{q}_{\text{max}}\notag\\
&\quad = \|\boldsymbol{t}\|_2^2 (1-\operatorname{Tr}(\boldsymbol{Q}))^2- \|\boldsymbol{t}\|^2_2(C_{\text{e}}+2\sqrt{N})^2C_{\text{e}}^2  \|\boldsymbol{Q}\|^2_2.\notag 
\end{align}
$f(\|\boldsymbol{t}\|_2\boldsymbol{q}_{\text{max}})\geq 0$ is equal to (\ref{eq76})$\geq 0$ and implies that $\operatorname{Tr}(\boldsymbol{Q}) + (C_{\text{e}}+2\sqrt{N})C_{\text{e}}\left\|\boldsymbol{Q}\right\|_2 -1\leq 0$, which is the constraint in (\ref{sdp}).

Finally, the constraint  $\operatorname{Tr}(\boldsymbol{Q}) + (C_{\text{e}}+2\sqrt{N})C_{\text{e}}\left\|\boldsymbol{Q}\right\|_2  -1 \leq 
	0$ in (\ref{sdp}) is equal to the semidefinite matrix condition in (\ref{eq72}), so the optimization problem (\ref{sdp}) is a convex SDP problem.

\section{The Proof for Proposition~\ref{prop2}}\label{ap2}
In the  dual problem (\ref{dual13}), with the  atomic norm definition having a gain-phase error $\boldsymbol{e}$ in (\ref{atomic}), the dual norm $\|\boldsymbol{U}\|^*_{\tilde{\mathcal{A}}}$ can be expressed as
\begin{align}\label{17}
\|\boldsymbol{U}\|^*_{\tilde{\mathcal{A}}} &= \sup_{\|\boldsymbol{X}\|_{\tilde{\mathcal{A}}}\leq 1} \langle \boldsymbol{X},\boldsymbol{U}\rangle\\
& \overset{(a)}{=} \sup_{\substack{ \|\boldsymbol{b}\|_1 \leq 1\\
		\theta_k\in[0,2\pi)\\
		\|\boldsymbol{d}_k\|_2\leq 1\\
		\|\boldsymbol{e}\|_2\leq C_{\text{e}}}}  \left \langle \sum_{k=0}^{K-1} b_k (\operatorname{diag}\{\boldsymbol{e}\}+\boldsymbol{I}) \boldsymbol{a}(\theta_k)\boldsymbol{d}_k^{\text{T}},\boldsymbol{U} 
\right\rangle\notag\\
& = \sup_{\substack{ \|\boldsymbol{b}\|_1 \leq 1\\
		\theta_k\in[0,2\pi)\\
		\|\boldsymbol{d}_k\|_2\leq 1\\
		\|\boldsymbol{e}\|_2\leq C_{\text{e}}}} \sum_{k=0}^{K-1} \left \langle  b_k (\operatorname{diag}\{\boldsymbol{e}\}+\boldsymbol{I}) \boldsymbol{a}(\theta_k)\boldsymbol{d}_k^{\text{T}},\boldsymbol{U} 
\right\rangle\notag\\
&  \overset{(b)}{=}  \sup_{\substack{ \|\boldsymbol{b}\|_1 \leq 1\\
		\theta_k\in[0,2\pi)\\
		\|\boldsymbol{d}_k\|_2\leq 1\\
		\|\boldsymbol{e}\|_2\leq C_{\text{e}}}} \sum_{k=0}^{K-1} 
\mathcal{R}\left\{ \boldsymbol{d}_k^{\text{T}}\boldsymbol{U}^{\text{H}}b_k (\operatorname{diag}\{\boldsymbol{e}\}+\boldsymbol{I}) \boldsymbol{a}(\theta_k) \right\}
\notag\\
&  \overset{(c)}{=} \sup_{\substack{ \|\boldsymbol{b}\|_1 \leq 1\\
		\theta_k\in[0,2\pi)\\ 
		\|\boldsymbol{e}\|_2\leq C_{\text{e}}}} \sum_{k=0}^{K-1} 
b_k \left\|  \boldsymbol{U}^{\text{H}} (\operatorname{diag}\{\boldsymbol{e}\}+\boldsymbol{I}) \boldsymbol{a}(\theta) 
\right\|_2
\notag\\
%& = \sup_{\substack{ \|\boldsymbol{b}\|_1 \leq 1\\
%		\theta_k\in[0,2\pi)\\ 
%		\|\boldsymbol{e}\|_2\leq C_{\text{e}}}} \left( \left\|  \boldsymbol{U}^{\text{H}} \operatorname{diag}\{\boldsymbol{e}\} \boldsymbol{a}(\theta) 
%+\boldsymbol{U}^{\text{H}} \boldsymbol{a}(\theta) 
%\right\|_2\right) \|\boldsymbol{b}\|_1 
%\notag \\
& \overset{(d)}{=} \sup_{\substack{ \|\boldsymbol{b}\|_1 \leq 1\\
		\theta_k\in[0,2\pi)\\ 
		\|\boldsymbol{e}\|_2\leq C_{\text{e}}}}  \left\|  \boldsymbol{U}^{\text{H}}  \left[\boldsymbol{e} 
+ \boldsymbol{a}(\theta) \right]
\right\|_2  \|\boldsymbol{b}\|_1 
\notag \\  
& = \sup_{\substack{ 
		\|\boldsymbol{e}\|_2\leq C_{\text{e}}
		\\
		\theta\in[0,2\pi) }}  \left\|  \boldsymbol{U}^{\text{H}}  \left[\boldsymbol{e} 
+ \boldsymbol{a}(\theta) \right]
\right\|_2, 
\notag 
\end{align}
where  $(a)$ is from the definition of atomic norm with gain-phase errors, $(b)$ is obtained with the definition of inner product between matrices, and $(c)$ is given by
\begin{align}
\boldsymbol{d}_k=\frac{\boldsymbol{U}^{\text{T}} (\operatorname{diag}\{\boldsymbol{e}\}^{\text{H}}+\boldsymbol{I}) \boldsymbol{a}^*(\theta_k)}{\left\|\boldsymbol{U}^{\text{T}} (\operatorname{diag}\{\boldsymbol{e}\}^{\text{H}}+\boldsymbol{I}) \boldsymbol{a}^*(\theta_k)\right\|_2}.
\end{align} 
For the equation (d), we formulate $\boldsymbol{e}'\triangleq \operatorname{diag}(\boldsymbol{e})\boldsymbol{a}(\theta) \in\mathbb{C}^{N\times 1}$, where  we use the steering vector $\boldsymbol{a}(\theta)$ ($\theta\in [0,2\pi)$) and a vector $\boldsymbol{e}\in\mathbb{C}^{N\times 1}$ ($\|\boldsymbol{e}\|_2\leq C_{\text{e}}$). Then, the $\ell_2$ norm of $\boldsymbol{e}'$ can be simplified as 	\begin{align}
		\|\boldsymbol{e}'\|^2_2 & =\|\operatorname{diag}(\boldsymbol{e})\boldsymbol{a}(\theta)\|^2_2 \\
		& =\boldsymbol{e}^{\text{H}}\operatorname{diag}(\boldsymbol{a}^*(\theta))\operatorname{diag}(\boldsymbol{a}(\theta)) \boldsymbol{e} =\|\boldsymbol{e}\|^2_2\leq C^2_{\text{e}}.\notag 
	\end{align}
	Therefore, we have 
	\begin{align}\label{eq8}
		& \sup_{\substack{  \theta_k\in[0,2\pi)\\ 
		\|\boldsymbol{e}\|_2\leq C_{\text{e}}}} \left\|  \boldsymbol{U}^{\text{H}} \underbrace{\operatorname{diag}\{\boldsymbol{e}\} \boldsymbol{a}(\theta) }_{\boldsymbol{e}'}
+\boldsymbol{U}^{\text{H}} \boldsymbol{a}(\theta) 
\right\|_2 \notag \\
&\quad  =\sup_{\substack{  \theta_k\in[0,2\pi)\\ 
		\|\boldsymbol{e}'\|_2\leq C_{\text{e}}}} \left\|  \boldsymbol{U}^{\text{H}} \left[\boldsymbol{e}'
+ \boldsymbol{a}(\theta) \right]
\right\|_2. 
	\end{align}
Then, the following equality in (d) can be obtained, where we just reuse the notation $\boldsymbol{e}$ instead of $\boldsymbol{e}'$ in (\ref{eq8}) to avoid introducing additional symbol $\boldsymbol{e}'$. 

Then, for the constraint  $\|\boldsymbol{U}\|^*_{\tilde{\mathcal{A}}}\leq \tau$, we build the following positive semidefinite matrix 
\begin{align}
\begin{bmatrix}
\boldsymbol{Q} & \boldsymbol{U}\\
\boldsymbol{U}^{\text{H}} & \boldsymbol{W}
\end{bmatrix}\succeq 0,
\end{align} 
where $\boldsymbol{Q}$ and $\boldsymbol{W}$ are the Hermitian matrices. With the Schur complement, when $\boldsymbol{W}$ is invertible, the matrix is positive semidefinite if and only if we have
\begin{align}
&\boldsymbol{Q}\succeq 0,\\
&\boldsymbol{Q}-\boldsymbol{U}\boldsymbol{W}^{-1}\boldsymbol{U}^{\text{H}} \succeq 0.
\end{align} 
Therefore, for any vector $\boldsymbol{t}\in\mathbb{C}^{N\times 1}$, we have $\boldsymbol{t}^\text{H}\boldsymbol{Q}\boldsymbol{t}\geq\boldsymbol{t}^\text{H} \boldsymbol{U}\boldsymbol{W}^{-1}\boldsymbol{U}^{\text{H}}\boldsymbol{t}$. By selecting $\boldsymbol{W}=\boldsymbol{\tau}^2\boldsymbol{I}_P$, we obtain 
\begin{align}\label{schur}
\left \| \boldsymbol{U}^{\text{H}}\boldsymbol{t}\right\|_2^2 \leq \tau^2\boldsymbol{t}^\text{H}\boldsymbol{Q}\boldsymbol{t}.
\end{align}  
When the gain-phase errors are considered, we formulate $\boldsymbol{t} = \boldsymbol{a}(\theta)+\boldsymbol{e}$, and we have
\begin{align}\label{21}
 &\left\|  \boldsymbol{U}^{\text{H}}  \left[\boldsymbol{e} 
+ \boldsymbol{a}(\theta) \right]
\right\|_2^2 \leq \tau^2 \left[\boldsymbol{e} 
+ \boldsymbol{a}(\theta) \right]^\text{H}\boldsymbol{Q}\left[\boldsymbol{e} 
+ \boldsymbol{a}(\theta) \right]\\
&\quad = \tau^2 \left( 
 \boldsymbol{e}^\text{H}\boldsymbol{Q}\boldsymbol{e} 
 +2\mathcal{R}\{ \boldsymbol{e}^\text{H}\boldsymbol{Q}\boldsymbol{a}(\theta) \}+  \boldsymbol{a}^\text{H}(\theta)\boldsymbol{Q} \boldsymbol{a}(\theta)  \right)\notag\\
 &\quad \leq \tau^2\left( C_{\text{e}}
 \frac{(\boldsymbol{Q}\boldsymbol{e} )^\text{H}}{\|\boldsymbol{Q}\boldsymbol{e} \|_2} \boldsymbol{Q}\boldsymbol{e} 
 +2\mathcal{R}\{ \boldsymbol{e}^\text{H}\boldsymbol{Q}\boldsymbol{a}(\theta) \}+  \boldsymbol{a}^\text{H}(\theta)\boldsymbol{Q} \boldsymbol{a}(\theta)  \right)\notag\\
 &\quad \leq  \tau^2 \left( C_{\text{e}}
 \|\boldsymbol{Q}\boldsymbol{e}\|_2 
 +2 C_{\text{e}}\|\boldsymbol{Q}\boldsymbol{a}(\theta) \|_2+  \boldsymbol{a}^\text{H}(\theta)\boldsymbol{Q} \boldsymbol{a}(\theta)  \right)\notag\\
  & \quad=  \tau^2 \left( C_{\text{e}}\sqrt{\|\boldsymbol{Q}\boldsymbol{e}\|^2_2 }
  +2 C_{\text{e}}\sqrt{\|\boldsymbol{Q}\boldsymbol{a}(\theta) \|_2^2}+  \boldsymbol{a}^\text{H}(\theta)\boldsymbol{Q} \boldsymbol{a}(\theta)  \right)\notag
\end{align}
Since $\|\boldsymbol{e}\|_2\leq C_{\text{e}}$ and $\boldsymbol{Q}\succeq 0$, we have $  \|\boldsymbol{Q}\boldsymbol{e}\|^2_2 \leq C^2_{\text{e}}\|\boldsymbol{Q}\|^2_2 $, where  $\|\boldsymbol{Q}\|_2\triangleq \lambda_{\max}(\boldsymbol{Q})$, and  $\lambda_{\max}(\boldsymbol{Q})$ is the largest singular value of $\boldsymbol{Q}$. Additionally, we have $\|\boldsymbol{Q}\boldsymbol{a}(\theta) \|_2^2 \leq N \|\boldsymbol{Q}\|^2 _2$.  Therefore, we can simplified (\ref{21}) as
\begin{align}
 \left\|  \boldsymbol{U}^{\text{H}}  \left[\boldsymbol{e} 
+ \boldsymbol{a}(\theta) \right]
\right\|_2^2 & \leq \tau^2\left[(C_{\text{e}}+2\sqrt{N})C_{\text{e}} \|\boldsymbol{Q}\|_2+  \boldsymbol{a}^\text{H}(\theta)\boldsymbol{Q} \boldsymbol{a}(\theta) \right]. 
\end{align} 
When $\boldsymbol{Q}$ satisfies the following condition 
\begin{align}\label{Q}
& \sum_n Q_{n,n+k}= 0 \ (k\neq 0)\\
& \operatorname{Tr}(\boldsymbol{Q})+(C_{\text{e}}+2\sqrt{N})C_{\text{e}} \left\|\boldsymbol{Q}\right\|_2 -1\leq 0, \notag
\end{align}
we have $ \left\|  \boldsymbol{U}^{\text{H}}  \left[\boldsymbol{e} 
+ \boldsymbol{a}(\theta) \right]
\right\|_2^2\leq \tau^2 $. Therefore, substitute into (\ref{17}) and the constraint $\|\boldsymbol{U}\|^*_{\tilde{\mathcal{A}}}\leq \tau$ is satisfied, then, the dual problem (\ref{dual13}) can be simplified as the SDP problem in (\ref{sdp}).

\section{The Proof for Proposition~\ref{prop4}} \label{ap3}
The entries of $\boldsymbol{N}$ follow the zero-mean Gaussian distribution with the variance being $\sigma_{\text{N}}^2$. The dual norm of $\boldsymbol{N}$ with the definition in (\ref{17}) can be simplified as 
\begin{align}
& \mathcal{E}\left\{\|\boldsymbol{N}\|^*_{\tilde{\mathcal{A}}}\right\}
 = \mathcal{E}\left\{ \sup_{\substack{ 
		\|\boldsymbol{e}\|_2\leq C_{\text{e}} 
		\\
		\theta\in[0,2\pi) }} \left\|  \boldsymbol{N}^{\text{H}}  \left[\boldsymbol{e} 
+ \boldsymbol{a}(\theta) \right]
\right\|_2\right\}\notag\\
&\quad \leq \mathcal{E}\left\{ \sup_{
		\|\boldsymbol{e}\|_2\leq C_{\text{e}} 
	} \left\|  \boldsymbol{N}^{\text{H}} \boldsymbol{e} 
\right\|_2\right\}+
\mathcal{E}\left\{ \sup_{ 
		\theta\in[0,2\pi) } \left\|  \boldsymbol{N}^{\text{H}}   \boldsymbol{a}(\theta) 
\right\|_2\right\}\notag\\ 
%&\quad = C_{\text{e}}  \mathcal{E}\left\{  \left\|  \boldsymbol{N} \right\|_2\right\}+
%\mathcal{E}\left\{ \sup_{ 
%	\theta\in[0,2\pi) } \left\|  \boldsymbol{N}^{\text{H}}   \boldsymbol{a}(\theta) 
%\right\|_2\right\}\notag\\ 
& \quad\leq C_{\text{e}}  \mathcal{E}\left\{  \left\|  \boldsymbol{N} \right\|_2\right\}+
\sigma_{\text{N}}\sqrt{4NP\ln N},
\end{align}
where $\mathcal{E}\left\{ \sup_{ 
	\theta\in[0,2\pi) } \left\|  \boldsymbol{N}^{\text{H}}   \boldsymbol{a}(\theta) 
\right\|_2\right\}\leq \sigma_{\text{N}}\sqrt{4NP\ln N}$ is obtained from Lemma 5.1 of \cite{DBLP:journals/corr/LiYTW17}. Additionally, we can obtain the upper bound of $\mathcal{E}\left\{  \left\|  \boldsymbol{N} \right\|_2\right\}$ as two types. The first type is formulated based on $\mathcal{E}\left\{  \left\|  \boldsymbol{N} \right\|_2\right\} \leq \mathcal{E}\left\{  \left\|  \boldsymbol{N} \right\|_F\right\}$. Since $ \left\|  \boldsymbol{N} \right\|_F$ follows chi distribution $\left\|  \boldsymbol{N} \right\|_F\sim\chi_{NP}$, we can obtain
\begin{align}
\mathcal{E}\left\{  \left\|  \boldsymbol{N} \right\|_F\right\} = \sqrt{2}\sigma_{\text{N}}\frac{\Gamma((NP+1)/2)}{\Gamma(NP/2)},
\end{align} 
where the gamma function is defined as $\Gamma(x)\triangleq \int^{\infty}_0 z^{x-1}e^{-z}\,dz$. Then, we have 
\begin{align}
\mathcal{E}\left\{  \left\|  \boldsymbol{N} \right\|_2\right\} \leq \mathcal{E}\left\{  \left\|  \boldsymbol{N} \right\|_F\right\} = \sqrt{2}\sigma_{\text{N}}\frac{\Gamma((NP+1)/2)}{\Gamma(NP/2)}\triangleq \text{bd}_1. 
\end{align}

With the probability more than $1-2e^{-t^2/2}$, the second type of upper bound can be formulated as 
\begin{align}
\mathcal{E}\left\{  \left\|  \boldsymbol{N} \right\|_2\right\} \leq \left(\sqrt{N}+\sqrt{P}+t\right)\sigma_{\text{N}}\triangleq \text{bd}_2,
\end{align}
where the upper bound is obtained from Theorem 5.35 of \cite{2010arXiv1011.3027V}.

Finally, the upper bound of $\mathcal{E}\left\{\|\boldsymbol{N}\|^*_{\tilde{\mathcal{A}}}\right\}$ can be obtained as
\begin{align}
\mathcal{E}\left\{\|\boldsymbol{N}\|^*_{\mathcal{A}}\right\}\leq \min\left\{\text{bd}_1,\text{bd}_2\right\}C_{\text{e}} +
\sigma_{\text{N}}\sqrt{4NP\ln N}.
\end{align}

\section{The entries of Fisher information matrix}\label{entry}

For the Fisher information $ \boldsymbol{F}=
\begin{bmatrix}
	\boldsymbol{F}_{1,1}&\boldsymbol{F}_{1,2}\\
	\boldsymbol{F}_{2,1}&\boldsymbol{F}_{2,2}
\end{bmatrix}$, the entries can be obtained as follows:
\begin{itemize}
	\item  The $k_1,k_2$-th entry of Fisher information matrix $\boldsymbol{F}_{1,1}$  can be obtained as 
	\begin{align}
	F^{1,1}_{k_1,k_2}
%	&=-\mathcal{E}\left\{\left.
%	\frac{\partial \ln f(\boldsymbol{y};\boldsymbol{\theta})}{\partial\theta_{k_1}\partial\theta_{k_2}}\right|\boldsymbol{\theta},\boldsymbol{g}
%	\right\}\notag\\
	& =  	\frac{\partial  
		\ln \det\{\boldsymbol{C}\}
	}{\partial\theta_{k_1}\partial\theta_{k_2}}  +
	\mathcal{E}\left\{ 	\frac{\partial  
		\boldsymbol{y}^{\text{H}}\boldsymbol{C}^{-1}\boldsymbol{y}
	}{\partial\theta_{k_1}\partial\theta_{k_2}} 
	\right\},
	\end{align}
	where the first term can be obtained as
	\begin{align}
	&
	\frac{\partial  
		\ln \det\{\boldsymbol{C}\}
	}{\partial\theta_{k_1}\partial\theta_{k_2}}=\operatorname{Tr}\left\{\frac{\partial 
		\boldsymbol{C}^{-1}\frac{\partial \boldsymbol{C}}{\partial \theta_{k_1}}  
	}{\partial\theta_{k_2}}\right\}\\
	&\qquad = \operatorname{Tr}\left\{\frac{ \partial 
		\boldsymbol{C}^{-1}
	}{\partial\theta_{k_2}} \frac{\partial \boldsymbol{C}}{\partial \theta_{k_1}} \right\}  
	+\operatorname{Tr}\left\{\boldsymbol{C}^{-1}   
	\frac{\partial \boldsymbol{C}}{\partial \theta_{k_1}\partial\theta_{k_2}}  
	\right\},\notag
	\end{align}
	and the second term is
	\begin{align}
	\mathcal{E}\left\{\frac{\partial  
		\boldsymbol{y}^{\text{H}}\boldsymbol{C}^{-1}\boldsymbol{y}
	}{\partial\theta_{k_1}\partial\theta_{k_2}}\right\} 
%& =
%	\operatorname{Tr}\left[
%	\frac{\partial  
%		\boldsymbol{C}^{-1}
%	}{\partial\theta_{k_1}\partial\theta_{k_2}}
%	\mathcal{E}\left\{ 
%	\boldsymbol{y}\boldsymbol{y}^{\text{H}}\right\} \right]\notag\\
	& = \operatorname{Tr}\left\{
	\frac{\partial  
		\boldsymbol{C}^{-1}
	}{\partial\theta_{k_1}\partial\theta_{k_2}}\boldsymbol{C}\right\}.
	\end{align}
	Therefore, $F^{1,1}_{k_1,k_2}$ can be simplified as
	\begin{align}
	F^{1,1}_{k_1,k_2} 
%	&= \operatorname{Tr}\Big\{\frac{ \partial 
%		\boldsymbol{C}^{-1}
%	}{\partial\theta_{k_2}} \frac{\partial \boldsymbol{C}}{\partial \theta_{k_1}}  
%	+ \boldsymbol{C}^{-1}   
%	\frac{\partial \boldsymbol{C}}{\partial \theta_{k_1}\partial\theta_{k_2}}  \notag\\
%	&\qquad
%	+ 	\frac{\partial  
%		\boldsymbol{C}^{-1}
%	}{\partial\theta_{k_1}\partial\theta_{k_2}}\boldsymbol{C}\Big\}\notag \\
	& =  \operatorname{Tr}\left\{\boldsymbol{C}^{-1}\frac{\partial  
		\boldsymbol{C}
	}{\partial\theta_{k_1}} \boldsymbol{C}^{-1} \frac{\partial  \boldsymbol{C}  	}{ \partial\theta_{k_2}}   
	\right\}.
	\end{align}

	\item The $k_1,k_2$-th entry of Fisher information matrix $\boldsymbol{F}_{1,2}$  can be obtained as 
	\begin{align}
	F^{1,2}_{k_1,k_2}
%	&=-\mathcal{E}\left\{\left.
%	\frac{\partial \ln f(\boldsymbol{y};\boldsymbol{\theta},\boldsymbol{g})}{\partial\theta_{k_1}\partial g_{k_2}}\right|\boldsymbol{\theta},\boldsymbol{g}
%	\right\}\notag\\
	& =  	\frac{\partial  
		\ln \det\{\boldsymbol{C}\}
	}{\partial\theta_{k_1}\partial g_{k_2}}  +
	\mathcal{E}\left\{ 	\frac{\partial  
		\boldsymbol{y}^{\text{H}}\boldsymbol{C}^{-1}\boldsymbol{y}
	}{\partial\theta_{k_1}\partial g_{k_2}} 
	\right\},
	\end{align}
	where the first term is simplified as 
	\begin{align}
	&\frac{\partial  
		\ln \det\{\boldsymbol{C}\}
	}{\partial\theta_{k_1}\partial g_{k_2}} =\operatorname{Tr}\left\{\frac{\partial 
		\boldsymbol{C}^{-1}\frac{\partial \boldsymbol{C}}{\partial \theta_{k_1}}  
	}{\partial g_{k_2}}\right\}\\
	&\quad = \operatorname{Tr}\left\{\frac{ \partial 
		\boldsymbol{C}^{-1}
	}{\partial g_{k_2}} \frac{\partial \boldsymbol{C}}{\partial \theta_{k_1}} \right\}  
	+\operatorname{Tr}\left\{\boldsymbol{C}^{-1}   
	\frac{\partial \boldsymbol{C}}{\partial \theta_{k_1}\partial g_{k_2}}  
	\right\},\notag
	\end{align}
	and the second term is
	\begin{align}
	\mathcal{E}\left\{\frac{\partial  
		\boldsymbol{y}^{\text{H}}\boldsymbol{C}^{-1}\boldsymbol{y}
	}{\partial\theta_{k_1}\partial g_{k_2}}\right\} 
%& =
%	\operatorname{Tr}\left[
%	\frac{\partial  
%		\boldsymbol{C}^{-1}
%	}{\partial\theta_{k_1}\partial g_{k_2}}
%	\mathcal{E}\left\{ 
%	\boldsymbol{y}\boldsymbol{y}^{\text{H}}\right\} \right]\notag\\
	& = \operatorname{Tr}\left\{
	\frac{\partial  
		\boldsymbol{C}^{-1}
	}{\partial\theta_{k_1}\partial g_{k_2}}\boldsymbol{C}\right\}.
	\end{align}
	Therefore, $F^{1,2}_{k_1,k_2}$ can be simplified as
	\begin{align}
	F^{1,2}_{k_1,k_2} 
%	&= \operatorname{Tr}\Big\{\frac{ \partial 
%		\boldsymbol{C}^{-1}
%	}{\partial g_{k_2}} \frac{\partial \boldsymbol{C}}{\partial \theta_{k_1}}  
%	+ \boldsymbol{C}^{-1}   
%	\frac{\partial \boldsymbol{C}}{\partial \theta_{k_1}\partial g_{k_2}}  \notag\\
%	&\qquad
%	+ 	\frac{\partial  
%		\boldsymbol{C}^{-1}
%	}{\partial\theta_{k_1}\partial g_{k_2}}\boldsymbol{C}\Big\}\notag \\ 
	&= \operatorname{Tr}\Bigg\{  
	\frac{\partial  
		\boldsymbol{C} 
	}{\partial\theta_{k_1} } \boldsymbol{C}^{-1}	\frac{\partial   
		\boldsymbol{C} 
	}{\partial  g_{k_2}} \boldsymbol{C}^{-1}
	\Bigg\}.
	\end{align}
	
	\item Similarly, we can get the $k_1,k_2$-th entry of Fisher information matrix $\boldsymbol{F}_{2,1}$  as
	\begin{align}
	F^{2,1}_{k_1,k_2}&=-\mathcal{E}\left\{\left.
	\frac{\partial \ln f(\boldsymbol{y};\boldsymbol{\theta},\boldsymbol{g})}{\partial g_{k_1}\partial \theta_{k_2}}\right|\boldsymbol{\theta},\boldsymbol{g}
	\right\}\notag\\
	& = \operatorname{Tr}\Bigg\{  
	\frac{\partial  
		\boldsymbol{C} 
	}{\partial g_{k_1} } \boldsymbol{C}^{-1}	\frac{\partial   
		\boldsymbol{C} 
	}{\partial  \theta_{k_2}} \boldsymbol{C}^{-1}
	\Bigg\}.
	\end{align}

	\item The $k_1,k_2$-th entry of Fisher information matrix $\boldsymbol{F}_{2,2}$  can be obtained as 
	\begin{align}
	F^{2,2}_{k_1,k_2}
%	&=-\mathcal{E}\left\{\left.
%	\frac{\partial \ln f(\boldsymbol{y};\boldsymbol{\theta})}{\partial g_{k_1}\partial g_{k_2}}\right|\boldsymbol{\theta},\boldsymbol{g}
%	\right\}\notag\\
	& =     \operatorname{Tr}\left\{\boldsymbol{C}^{-1}\frac{\partial  
		\boldsymbol{C}
	}{\partial g_{k_1}} \boldsymbol{C}^{-1} \frac{\partial  \boldsymbol{C}  	}{ \partial g_{k_2}}   
	\right\}.
	\end{align}
\end{itemize}
	
From the entries of block matrices $\boldsymbol{F}_{1,1},\boldsymbol{F}_{1,2}$, $\boldsymbol{F}_{2,1}$ and $\boldsymbol{F}_{2,2}$, we can find that the expressions of $\frac{\partial  
	\boldsymbol{C}
}{\partial\theta_{k}}$ and $\frac{\partial  
\boldsymbol{C}
}{\partial g_{k}}$ must be calculated, so we can obtain the expressions as follows:
\begin{itemize}
	\item For $\frac{\partial  
		\boldsymbol{C}
	}{\partial\theta_{k}}$,  we have
\begin{align}
\frac{\partial  
\boldsymbol{C}
}{\partial\theta_{k}} 
%= & \frac{\partial  
%(\boldsymbol{I}\otimes \boldsymbol{GA})\boldsymbol{B}(\boldsymbol{I}\otimes (\boldsymbol{GA})^{\text{H}})
%}{\partial\theta_{k}}\notag\\
= & \left(\boldsymbol{I}\otimes \boldsymbol{G}\frac{\partial  
\boldsymbol{A}
}{\partial\theta_{k}}\right)\boldsymbol{B}(\boldsymbol{I}\otimes (\boldsymbol{GA})^{\text{H}})\notag\\
&+(\boldsymbol{I}\otimes \boldsymbol{GA})\boldsymbol{B} \left( \boldsymbol{I}\otimes\frac{\partial  
\boldsymbol{A}^{\text{H}}
}{\partial\theta_{k}}\boldsymbol{G}^{\text{H}}\right)
\end{align}
where $\frac{\partial  
	\boldsymbol{A}
}{\partial\theta_{k}}$ is expressed as $
\frac{\partial  
\boldsymbol{A}
}{\partial\theta_{k}} = \begin{bmatrix}
\boldsymbol{0},	\frac{\partial  
\boldsymbol{a}(\theta_k)
}{\partial\theta_{k}},\boldsymbol{0}
\end{bmatrix}$, and $\frac{\partial  
\boldsymbol{a}(\theta_k)
}{\partial\theta_{k}}$ can be obtained easily.
\item For $\frac{\partial  
	\boldsymbol{C}
}{\partial g_{k}}$, we can obtain
\begin{align}
\frac{\partial  
	\boldsymbol{C}
}{\partial g_{k}} 
%& =\frac{\partial  
%	(\boldsymbol{I}\otimes \boldsymbol{GA})\boldsymbol{B}(\boldsymbol{I}\otimes (\boldsymbol{GA})^{\text{H}})
%}{\partial g_{k}}\\
& = \left(\boldsymbol{I}\otimes \frac{\partial  
	\boldsymbol{G}
}{\partial g_{k}} \boldsymbol{A}\right)\boldsymbol{B}(\boldsymbol{I}\otimes (\boldsymbol{GA})^{\text{H}}), \notag
\end{align}
where $\frac{\partial  
	\boldsymbol{G}
}{\partial g_{k}}$ can be obtained easily. 
\end{itemize}

\bibliographystyle{IEEEtran}
\bibliography{IEEEabrv.bib,References.bib}

\begin{IEEEbiography}[{\includegraphics[width=1in,height=1.25in,clip,keepaspectratio]{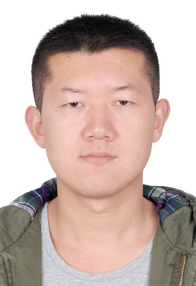}}]{Peng Chen (S'15-M'17)}
	was born in Jiangsu, China in 1989. He received the B.E. degree in 2011 and the Ph.D. degree in 2017, both from the School of Information Science and Engineering, Southeast University, China. From Mar. 2015 to Apr. 2016, he was a Visiting Scholar in the Electrical Engineering Department, Columbia University, New York, NY, USA.
	
	He is now an associate professor at the State Key Laboratory of Millimeter Waves, Southeast University. His research interests include radar signal processing and millimeter wave communication.
	
\end{IEEEbiography}

\begin{IEEEbiography}[{\includegraphics[width=1in,height=1.25in,clip,keepaspectratio]{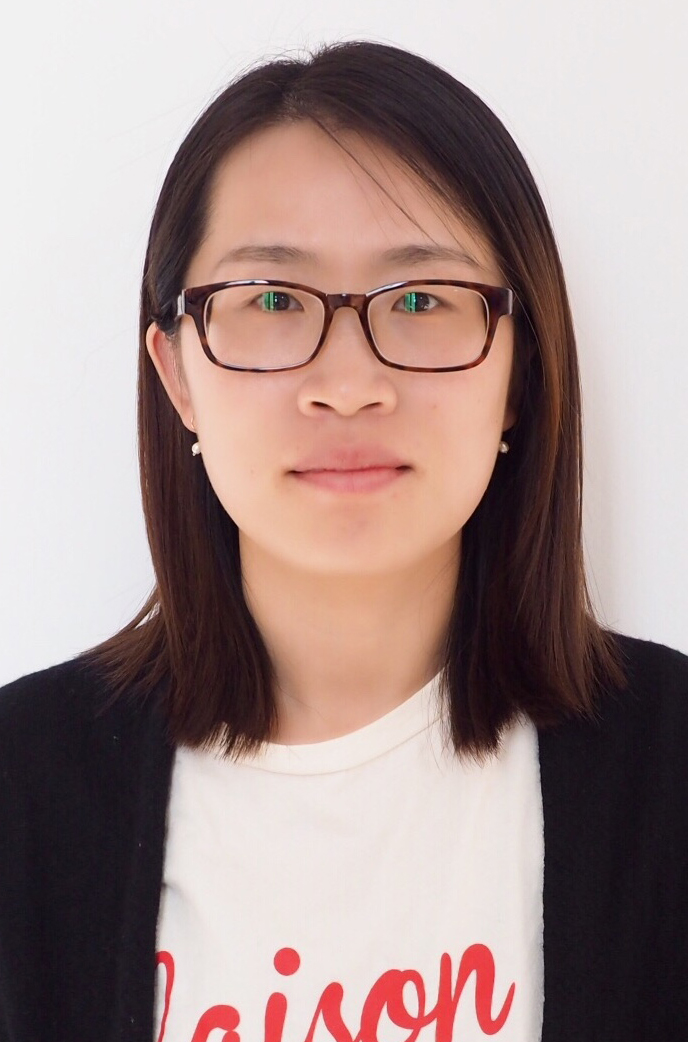}}]{Zhimin Chen (M'17)}  was born in Shandong, China, in 1985. She received the Ph.D. degree in information and communication engineering from the School of Information Science and Engineering, Southeast University, Nanjing, China in 2015. Since 2015, she has been with Shanghai Dianji University, Shanghai, China, where she is a Professor. Her research interests include array signal processing, vehicle communications and millimeter-wave communications.
\end{IEEEbiography}

\begin{IEEEbiography}[{\includegraphics[width=1in,height=1.25in,clip,keepaspectratio]{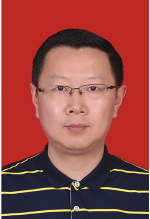}}]{Zhenxin Cao (M'18)}
	was born in May 1976. He received the M. S. degree  in 2002 from Nanjing University of Aeronautics and Astronautics, China, and the Ph.D. degree in 2005 from  the School of Information Science and Engineering, Southeast University, China. From 2012 to 2013, he was a Visiting Scholar in North Carolina State University. 
	
	Since 2005, he has been with the State Key Laboratory of Millimeter Waves, Southeast University, where he is a Professor. His research interests include antenna theory and application. 
	
\end{IEEEbiography}

\begin{IEEEbiography}[{\includegraphics[width=1in,height=1.25in,clip,keepaspectratio]{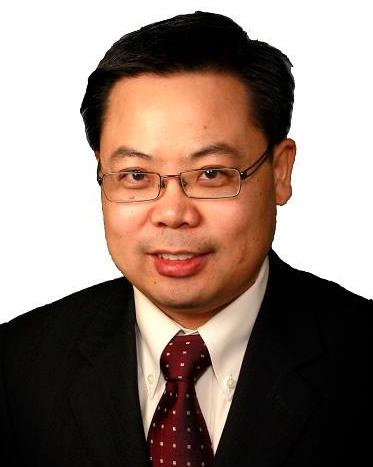}}]{Xianbin Wang (S'98-M'99-SM'06-F'17)} is a Professor and Tier-I Canada Research Chair at Western University, Canada. He received his Ph.D. degree in electrical and computer engineering from National University of Singapore in 2001.
	
	Prior to joining Western, he was with Communications Research Centre Canada (CRC) as a Research Scientist/Senior Research Scientist between July 2002 and Dec. 2007. From Jan. 2001 to July 2002, he was a system designer at STMicroelectronics, where he was responsible for the system design of DSL and Gigabit Ethernet chipsets.  His current research interests include 5G technologies, Internet-of-Things, communications security, machine learning and locationing technologies. Dr. Wang has over 300 peer-reviewed journal and conference papers, in addition to 26 granted and pending patents and several standard contributions.
	
	Dr. Wang is a Fellow of Canadian Academy of Engineering, a Fellow of IEEE and an IEEE Distinguished Lecturer. He has received many awards and recognitions, including Canada Research Chair, CRC President’s Excellence Award, Canadian Federal Government Public Service Award, Ontario Early Researcher Award and five IEEE Best Paper Awards. He currently serves as an Editor/Associate Editor for IEEE Transactions on Communications, IEEE Transactions on Broadcasting, and IEEE Transactions on Vehicular Technology and He was also an Associate Editor for IEEE Transactions on Wireless Communications between 2007 and 2011, and IEEE Wireless Communications Letters between 2011 and 2016. Dr. Wang was involved in many IEEE conferences including GLOBECOM, ICC, VTC, PIMRC, WCNC and CWIT, in different roles such as symposium chair, tutorial instructor, track chair, session chair and TPC co-chair.
	
\end{IEEEbiography}

\end{document}